\documentclass[%
 reprint,
 superscriptaddress,
nofootinbib,
 amsmath,amssymb,
 aps,
]{revtex4-2}

\usepackage[utf8]{inputenc}
\usepackage[english]{babel} 
\usepackage{graphicx}
\usepackage{dcolumn}
\usepackage{bm}
\usepackage[bookmarks, colorlinks=true, urlcolor=blue, linkcolor=black, citecolor=blue]{hyperref}
\usepackage{mathtools}
\usepackage{physics}
\usepackage{enumitem}
\usepackage[dvipsnames]{xcolor}
\usepackage{float}
\usepackage[capitalise]{cleveref}
\usepackage{amsthm}
\newtheorem{thm}{Theorem}
\newtheorem{lem}[thm]{Lemma}
\newtheorem{dfn}[thm]{Definition}

\newtheorem{theorem}{Theorem}[section]
\newtheorem{lemma}[theorem]{Lemma}
\newtheorem{definition}[theorem]{Definition}
\newtheorem{corollary}[theorem]{Corollary}
\usepackage[algo2e, linesnumbered, ruled, vlined]{algorithm2e}
\usepackage{etoolbox}

\usepackage{float}
\usepackage{physics}

\newcommand{\x}{\ensuremath{\bm{x}}}
\newcommand{\y}{\ensuremath{\bm{y}}}
\newcommand{\thv}{\ensuremath{\bm{\theta}}}
\renewcommand{\O}{\ensuremath{\mathcal{O}}}
\newcommand{\X}{\ensuremath{\mathcal{X}}}
\newcommand{\eps}{\ensuremath{\varepsilon}}
\newcommand{\A}{\ensuremath{\mathcal{A}}}
\newcommand{\C}{\ensuremath{\mathcal{C}}}
\newcommand{\D}{\ensuremath{\mathcal{D}}}
\newcommand{\E}{\ensuremath{\mathbb{E}}}
\newcommand{\F}{\ensuremath{\mathcal{F}}}
\newcommand{\R}{\ensuremath{\mathbb{R}}}
\newcommand{\Z}{\ensuremath{\mathbb{Z}_{N}}}

\allowdisplaybreaks

\makeatletter 

\renewcommand\onecolumngrid{
\do@columngrid{one}{\@ne}%
\def\set@footnotewidth{\onecolumngrid}
\def\footnoterule{\kern-6pt\hrule width 1.5in\kern6pt}%
}

\begin{document}

\preprint{APS/123-QED}

\title{Shadows of quantum machine learning}

\author{Sofiene Jerbi}
\affiliation{Institute for Theoretical Physics, University of Innsbruck, Austria}
\affiliation{Dahlem Center for Complex Quantum Systems, Freie Universit\"{a}t Berlin, Germany}
\author{Casper Gyurik}
\affiliation{applied Quantum algorithms (aQa), Leiden University, The Netherlands}
\author{Simon C.~Marshall}
\affiliation{applied Quantum algorithms (aQa), Leiden University, The Netherlands}
\author{Riccardo Molteni}
\affiliation{applied Quantum algorithms (aQa), Leiden University, The Netherlands}
\author{Vedran Dunjko}
\affiliation{applied Quantum algorithms (aQa), Leiden University, The Netherlands}


\date{\today}

\begin{abstract}
\noindent Quantum machine learning is often highlighted as one of the most promising practical applications for which quantum computers could provide a computational advantage. However, a major obstacle to the widespread use of quantum machine learning models in practice is that these models, even once trained, still require access to a quantum computer in order to be evaluated on new data.\break To solve this issue, we introduce a new class of quantum models where quantum resources are only required during training, while the deployment of the trained model is classical. Specifically, the training phase of our models ends with the generation of a `shadow model' from which the classical deployment becomes possible. We prove that: \emph{i)} this class of models is universal for classically-deployed quantum machine learning; \emph{ii)} it does have restricted learning capacities compared to `fully quantum' models, but nonetheless \emph{iii)} it achieves a provable learning advantage over fully classical learners, contingent on widely-believed assumptions in complexity theory. These results provide compelling evidence that quantum machine learning can confer learning advantages across a substantially broader range of scenarios, where quantum computers are exclusively employed during the training phase. By enabling classical deployment, our approach facilitates the implementation of quantum machine learning models in various practical contexts.
\end{abstract}

\maketitle


\section{Introduction\label{sec:intro}}


\noindent Quantum machine learning is a rapidly growing field \cite{biamonte17,dunjko18,schuld18}\break driven by its potential to achieve quantum advantages in practical applications. A particularly interesting approach to make quantum machine learning applicable in the near term is to develop learning models based on parametrized quantum circuits \cite{benedetti19,cerezo21,bharti21}. Indeed, such quantum models have already been shown to achieve good learning performance in benchmarking tasks, both in numerical simulations \cite{schuld19b,schuld20b,liu18,jerbi21,skolik21} and on actual quantum hardware \cite{havlivcek19,zhu19,peters21,haug21}. Moreover, based on widely-believed cryptography assumptions, these models also hold the promise to solve certain learning tasks that are intractable for classical algorithms \cite{liu20,gyurik22}, including predicting ground state properties of highly-interacting quantum systems \cite{gyurik23b}.

Despite these advances, quantum machine learning is facing a major obstacle for its use in practice. A typical workflow of a machine learning model involved, e.g.,\break in driving autonomous vehicles, is divided into: (i) a \emph{training phase}, where the model is trained, typically using training data or by reinforcement; followed by (ii) a \emph{deployment phase}, where the trained model is evaluated on new input data. For quantum machine learning models, both of these phases require access to a quantum computer. But given that in many practical machine learning applications, the trained model is meant for a widespread deployment, the current scarcity of quantum computing access dramatically reduces the applicability of quantum machine learning. One way of addressing this problem is by generating \emph{shadow models} out of quantum machine learning models. That is, we propose inserting a \emph{shadowing phase} between the training and deployment, where a quantum computer is used to collect information on the quantum model. Then a classical computer can use this information to evaluate the model on new data during the deployment phase.

The conceptual idea of generating shadows of quantum models was already proposed by Schreiber \emph{et al.}~\cite{schreiber22}, albeit under the terminology of \emph{classical surrogates}. In that work, as well as in that of Landman \emph{et al.}~\cite{landman22}, the authors make use of the general expression of quantum models as trigonometric polynomials \cite{schuld21b} to learn the Fourier representation of trained models and evaluate them classically on new data. However, these works also suggest that a classical model could potentially be trained directly on the training data and achieve the same performance as the shadow model, thus circumventing the need for a quantum model in the first place. This raises the concern that all quantum models that are compatible with a classical deployment would also lose all quantum advantage, hence severely limiting the prospects for a widespread use of quantum machine learning.\\
Therefore, two natural open questions are raised: 
\begin{enumerate}[leftmargin=6mm]
    \item \emph{Can shadow models achieve a quantum advantage over entirely classical (classically trained and classically evaluated) models?}
    \item \emph{Do there exist quantum models that do not admit efficiently evaluatable shadow models?}
\end{enumerate}

\begin{figure*}
	\centering
	\includegraphics[width=\linewidth]{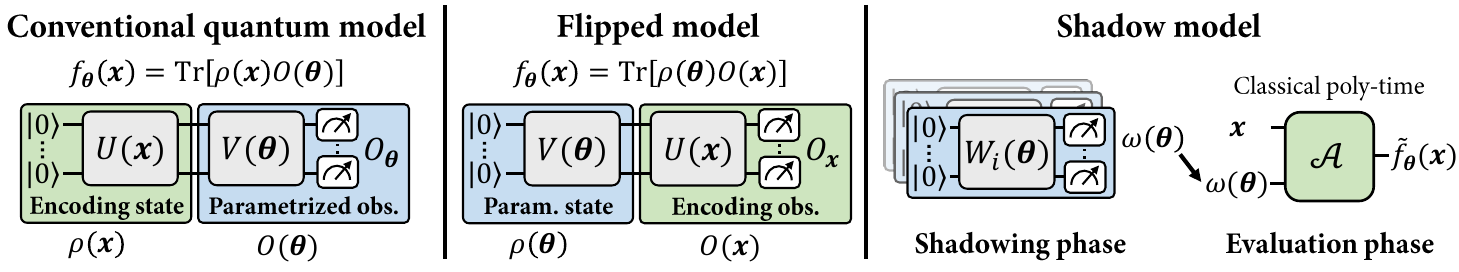}
    \vspace{-1.5em}
	\caption{\textbf{Quantum and shadow models.} (left) Conventional quantum models can be expressed as inner products between a data-encoding quantum state $\rho(\x)$ and a parametrized observable $O(\thv)$. The resulting linear model $f_{\thv}(\x)= \Tr[\rho(\x)O(\thv)]$ naturally corresponds to a quantum computation, depicted here. (middle) We define flipped models $f_{\thv}(\x)= \Tr[\rho(\thv)O(\x)]$ as quantum linear models where the role of the quantum state $\rho(\thv)$ and the observable $O(\x)$ is flipped compared to conventional models. (right) Flipped models are associated to natural shadow models: one can use techniques from shadow tomography to construct a classical representation $\hat{\rho}(\thv)$ of the parametrized state $\rho(\thv)$ (during the shadowing phase), such that, for encoding observables $O(\x)$ that are classically representable (e.g., linear combinations of Pauli observables), $\hat{\rho}(\thv)$ can be used by a classical algorithm to evaluate the model $f_{\thv}(\x)$ on new input data (during the evaluation phase). More generally, a shadow model is defined by (i) a shadowing phase where a (bit-string) advice $\omega(\thv)$ is generated by the evaluation of multiple quantum circuits $W_1(\thv), \ldots, W_M(\thv)$, and (ii) an evaluation phase where this advice is used by a classical algorithm $\A$, along with new input data $\x$ to evaluate their labels $\widetilde{f}_{\thv}(\x)$. In \cref{sec:general-shadow}, we show that under this general definition, all shadow models are shadows of flipped models.}
     \vspace{-0.5em}
	\label{fig:main}
\end{figure*}

\noindent In this work, we resolve both of these key open questions.\break We propose a general definition for shadow models, rooted in the fundamental idea that quantum machine learning models can be universally expressed as linear models \cite{jerbi21b}. This formulation of shadow models allows us to leverage various results and techniques from quantum information theory for the analysis of this model class. From a practical perspective, employing shadow tomography techniques \cite{huang20b,bertoni22,wan23,hu23} allows to easily construct diverse shadow models that will resonate with the practitioners of quantum machine learning. Furthermore, in our exploration of the computational capabilities of shadow models, we find them to capture a distinct computational class. Specifically, we demonstrate that, under widely-believed cryptography assumptions, there exist learning tasks where shadow models exhibit a provable quantum advantage over fully classical models. However, contrary to this advantage, we also establish that there exist quantum models that are strictly more powerful than the class of shadow models, based on common assumptions in complexity theory. 

For ease of exposition, we will first adhere to a working definition of a shadow model as a model that is trained on a quantum computer, but can be evaluated classically on new input data with the help of information generated by a quantum computer (i.e., quantum-generated advice) that is independent of the new data. We will (informally) call a model ``shadowfiable" if there exists a method of turning it into a shadow model. In \cref{sec:general-shadow}, we will make our definitions more precise.

\section{The flipped model}

The construction of our shadow models starts from a simple yet key observation: all standard quantum machine learning models for supervised learning can be expressed as linear models \cite{jerbi21b}. To delve into this claim, we first draw upon early works that utilized parametrized quantum circuits in machine learning \cite{havlivcek19,schuld19b}. These works proposed quantum models that are naturally expressed as linear functions of the form
\begin{equation}\label{eq:linear-models}
f_{\thv}(\x)  = \Tr[\rho(\x)O(\thv)]
\end{equation}
where $\rho(\x)$ are quantum states that encode classical data $\x\in\X$ and $O(\thv)$ are parametrized observables whose inner product with $\rho(\x)$ defines $f_{\thv}(\x)$ (see Fig.~\ref{fig:main}). In a regression task, one would use such a model to assign a real-valued label to an input $\x$, while in classification tasks, one would additionally apply, e.g., a sign function, to discretize its output into a class. From a circuit picture, such models can be evaluated on a quantum computer by:\break (i) preparing an initial state $\rho_0$, e.g., $\ket{0}\!\bra{0}^{\otimes n}$, (ii) evolving it under a data-dependent circuit $U(\x)$, (iii) followed by a variational circuit $V(\thv)$, (iv) before finally measuring the expectation value of a Hermitian observable $O$. Together, steps (i) and (ii) define
\begin{equation}
\rho(\x)=U(\x)\rho_0U^\dagger(\x),
\end{equation}
while steps (iii) and (iv) define
\begin{equation}\label{eq:obs}
O(\thv)=V^\dagger(\thv)OV(\thv).
\end{equation}
Since the early works, it is known that quantum linear models also capture quantum kernel models as a special case \cite{schuld21}, simply by making $O(\thv)$ directly dependent on the training data of the learning task. Perhaps more surprisingly, quantum linear models can also encompass more general data re-uploading models, composed of several layers of data encoding and variational processing $U_1(\x)V_1(\thv)U_2(\x)\ldots$ Indeed, data re-uploading models can be mapped to linear models through circuit transformations (e.g., gate teleportation) that relocate all data-encoding gates to the first layer of the circuit \cite{jerbi21b}.

\subsection{Flipped model definition}

The definition of a quantum linear model in \cref{eq:linear-models} can in general accommodate any pair of Hermitian operators in place of $\rho(\x),O(\thv)$. However, due to how these models are evaluated on a quantum computer, one commonly works under the constraint that $\rho(\x)$ defines a quantum state (i.e., a positive semi-definite operator with unit trace). Indeed, from an operational perspective, $\rho(\x)$ must be physically prepared on a quantum device before being measured with respect to the observable $O(\thv)$ (which only needs to be Hermitian in order to be a valid observable).

For reasons that will become clearer from the shadowing perspective, we define a so-called \emph{flipped model}, where we flip the role of $\rho(\x)$ and $O(\thv)$. That is, we consider
\begin{equation}\label{eq:flipped-model}
f_{\thv}(\x)  = \Tr[\rho(\thv)O(\x)]
\end{equation}
where $\rho(\thv)$ is a parametrized quantum state and $O(\x)$ is an observable that encodes the data and can take more general forms than \cref{eq:obs} as we will see next. This model also corresponds to a straightforward quantum computation as $\rho(\thv)$ can be physically prepared before being measured with respect to $O(\x)$.

A simple example of flipped model is for instance defined by:
\begin{equation}\label{eq:Pauli-shadow-model}
\rho(\thv) = V(\thv) \rho_0 V^\dagger(\thv)\quad \&\quad O(\x) = \sum_{j=1}^m w_j(\x) P_j
\end{equation}
for an initial state $\rho_0$, a variational circuit $V(\thv)$, and a collection of Pauli observables $\{P_j\}_{j=1}^m$ weighted by data-dependent weights $w_j(\x)\in\R$. One can evaluate this model by repeatedly preparing $\rho(\thv)$ on a quantum computer, measuring it in a Pauli basis specified by a $P_j$, and weighting the outcome by $w_j(\x)$. For other examples of flipped models, see Appendix \ref{appdx:shadow-models}.

As opposed to conventional quantum linear models, flipped models are well-suited to construct shadow models. Since the variational operators $\rho(\thv)$ are quantum states, one can straightforwardly use techniques from shadow tomography \cite{huang20b} to construct classical shadows $\hat{\rho}(\thv)$ of these states. What we call classical shadows $\hat{\rho}(\thv)$ here are collections of measurement outcomes obtained from copies of $\rho(\thv)$ that can be used to classically approximate expectation values of certain observables $O$ (for a certain restricted family). If we take these observables to be our data-dependent $O(\x)$, then we end up with a classical model $\widetilde{f}_{\thv}(\x)$ that approximates our flipped model. Note here that one has total freedom on the classical shadow techniques they may use to define their shadow models, and a plethora of protocols have already been proposed in the literature \cite{huang20b,bertoni22,wan23,hu23}. But it is important to keep in mind that each of these protocols comes with its limitations, as it may restrict the class of states $\rho(\thv)$ or the class of observables $O(\x)$ for which an \emph{efficient and faithful} shadow model can be constructed. By \emph{efficient} we refer here to the number of measurements performed on $\rho(\thv)$ and the time complexity of estimating the expectation values of observables $O(\x)$ from these measurements. And by \emph{faithful} we refer to the approximation error between the shadow model $\widetilde{f}_{\thv}(\x)$ resulting from the shadow protocol and the original flipped model $f_{\thv}(\x)$. For instance, in the example of \cref{eq:Pauli-shadow-model}, we know that if all Pauli operators $\{P_j\}_{j=1}^m$ are $k$-local,\break then $\widetilde{\O}\left(3^kB^2\eps^{-2}\right)$ measurements of $\rho(\thv)$, where $B~=~\max_{\x} \sum_{j=1}^m \abs{w_i(\x)}$, are sufficient to guarantee $\max_{\x} \big|\widetilde{f}_{\thv}(\x)-f_{\thv}(\x)\big|\leq\eps$ with high probability. But for non-local Pauli operators (i.e., large $k$), this protocol becomes highly inefficient if we want to guarantee a low error $\eps$.

Importantly, shadowfied flipped models are not limited to constructions based on classical shadow protocols. Given that the states $\rho(\thv)$ are not given to us a black-box (as is generally assumed in shadow tomography), one can use prior knowledge on these states to construct efficient shadowing procedure. For instance, if $\rho(\thv)$ is known to be a superposition of a tractable number of computational basis states, or well-approximated by a matrix product state (MPS) with low bond dimension, then efficient tomography protocols may be used \cite{cramer10}.

\subsection{Properties of flipped models}

Flipped models are a stepping stone toward the claims of quantum advantage and ``shadowfiability'' that are the focus of this paper. Nonetheless, they constitute a newly introduced model, which is why it is useful to understand first how they relate to previous quantum models and what learning guarantees they can have.

Since conventional linear models of the form of \cref{eq:linear-models} play a central role in quantum machine learning, we start by asking the question: when can these models be represented by (efficiently evaluatable) flipped models? That is, given a conventional model $f_{\thv}(\x)~=~\Tr[\rho(\x)O(\thv)]$, can we construct a flipped model $\widetilde{f}_{\thv}(\x)~=~\Tr[\rho'(\thv)O'(\x)]$ such that $\widetilde{f}_{\thv}(\x)~\approx~f_{\thv}(\x), \forall \x,\thv$, and $\widetilde{f}_{\thv}(\x)$ is as efficient to evaluate as $f_{\thv}(\x)$. Clearly, a conventional model $f_{\thv}(\x)$ for which the parametrized operator $O(\thv)$ is also a quantum state (i.e., a positive semi-definite trace-1 operator) is by definition also a flipped model. Therefore, a natural strategy to flip a conventional model is to transform its observable $O(\thv)$ into a quantum state $\rho'(\thv)$. This transformation involves dealing with the negative eigenvalues of $O(\thv)$, which is straightforward\footnote{The sign of the eigenvalues of the observable $O(\thv)$ can be taken into account using an auxiliary qubit, without overheads in the efficiency of evaluation. See Appendix \ref{appdx:flipping-bounds} for more details.}, as well as \emph{normalizing} these eigenvalues, which more importantly affects the efficiency of evaluating the resulting flipped model. Indeed, the normalization factor $\alpha$ that results from normalizing $O(\thv)$ corresponds to its trace norm $\norm{O}_1~=~\Tr\big[\sqrt{O^2}\big]$ and needs to be absorbed into the observable $O'(\x)~=~\alpha\rho(\x)$ of the flipped model $\widetilde{f}_{\thv}(\x)$ to guarantee $\widetilde{f}_{\thv}(\x)~=~f_{\thv}(\x)$. This directly impacts the spectral norm $\norm{O'}_\infty = \max_{\ket{\psi}} \expval{O'}_\psi = \alpha$ of the flipped model, and therefore the efficiency of its evaluation, as $\O(\norm{O'}_\infty^2/\eps^2)$ measurements of $\rho'(\thv)$ are needed in order to estimate $\widetilde{f}_{\thv}(\x)$ to additive error $\eps$ (see Appendix \ref{appdx:sample-complexity} for a derivation). Therefore, we end up showing that, for a conventional model $f_{\thv}(\x)$ acting on $n$ qubits and with a bounded observable trace norm $\norm{O}_1 \leq \alpha$, we can construct a flipped model acting on $m=n+1$ qubits and with observable spectral norm $\norm{O'}_\infty = \alpha$.

Interestingly, in the relevant regime where the number of qubits $n,m$ used by the linear models involved in this flipping is logarithmic in $\norm{O}_1$ (e.g., where $O$ is a Pauli observable and hence $\norm{O}_1=2^n$), we find that this requirement on the spectral norm $\norm{O'}_\infty$ of the resulting flipped model is unavoidable in the worst case, up to a logarithmic factor in $\norm{O}_1$. We refer to Appendix \ref{appdx:flipping-bounds} for proof of these statements and a more in-depth discussion.

Another property of interest in machine learning is the generalization performance of a learning model. That is, we want to bound the gap between the performance of the model on its training set (so-called training error) and its performance on the rest of the data space (or expected error). Such bounds have for instance been derived in terms of the number of encoding gates in the quantum model \cite{caro21}, or the rank of its observable \cite{gyurik23}. In the case of flipped model, we find instead a bound in terms of the number of qubits $n$ and the spectral norm $\norm{O}_\infty$ of the observable. Since these quantities are operationally meaningful, this gives us a natural way of controlling the generalization performance of our flipped models. Stated informally, we find that if a flipped model achieves a small error $\abs{f_{\thv}(\x) -f(\x)}\leq\eta$ for all $\x$ in a training set of size $M$, then we only need $M$ to scale as $\widetilde{\Omega}\left(\frac{n\norm{O}_\infty^2}{\eps\eta^2}\right)$ in order to guarantee a small expected error $\abs{f_{\thv}(\x) -f(\x)}\leq2\eta$ with probability $1-\eps$ over the entire data distribution.

Note that the dependence on $n$ and $\norm{O}_\infty$ is linear and quadratic, respectively, which means that we can afford a large number of qubits and a large spectral norm and still guarantee a good generalization performance. This is particularly relevant as the spectral norm is a controllable quantity, meaning we can easily fine-tune our models to perform well in training and generalize well. E.g., in the case of the model in \cref{eq:Pauli-shadow-model}, this spectral norm is bounded by $\max_{\x} \sum_{j=1}^m \abs{w_i(\x)}$, which scales favorably with the number of qubits $n$ if $m\in\O(\text{poly}(n))$ or if the vector $\bm{w}(\x)$ is sparse.

\subsection{Quantum advantage of a shadow model}

We recall that we (informally) define shadow models as models that are trained on a quantum computer, but, after a shadowing procedure that collects information on the trained model, are evaluated classically on new input data. In this section, we consider the question of achieving a quantum advantage using such shadow models. It may seem at first sight that this question has a straightforward answer, which is ``no": if the function learned by a model is classically computable, then there should be no room for a quantum advantage. However, as demonstrated in Refs.~\cite{gyurik22,servedio04}, one can also achieve a quantum advantage based on so-called \emph{trap-door functions}.\break These are functions that are believed to be hard to compute classically, unless given a key (or advice) that allows for an efficient classical computation. Notably, there exist trap-door functions where this key can be efficiently computed using a quantum computer, but not classically. This allows us to construct shadow models that make use of this quantum-generated key to compute an otherwise classically untractable function. 

Similarly to related results showing a quantum advantage in machine learning with classical data \cite{liu20,sweke21}, we consider a learning task where the target function (i.e., the function generating the training data) is derived from cryptographic functions that are widely believed to be\break hard to compute classically. More precisely, we introduce a variant of the discrete cube root learning task~\cite{gyurik22}, which is hard to solve classically under a hardness assumption related to that of the RSA cryptosystem~\cite{kearns94}. In this task, we consider target functions defined on $\Z=\{0,\ldots, N-1\}$~as
\begin{equation}
	g_s(\x) = \begin{cases}
    1,&\text{if }\sqrt[3]{\x} \text{ mod } N \in[s,s+\frac{N-1}{2}],\\ 0, &\text{otherwise}
    \end{cases}
\end{equation}
where $N=pq$ is an $n$-bit integer, product of two primes $p,\ q$ of the form $3k+2,\ 3k'+2$, such that the discrete cube root is properly defined as the inverse of the function $\y^3 \text{ mod } N$. These target functions are particularly appealing because of a number of interesting properties:
\begin{enumerate}[label=(\roman*),leftmargin=8mm]
    \item It is believed that given only $\x$ and $N$ as input, computing $g(\x)=\sqrt[3]{\x} \text{ mod } N$ with high probability of success over random draws of $\x$ and $N$ is classically intractable. This assumption is known as the discrete cube root (DCR) assumption.
    \item On the other hand, computing $\x^a \text{ mod } N$ is classically efficient for any $a\in \Z$. For $a=3$, this implies that $g^{-1}(\y) = \y^3 \text{ mod } N$ is a one-way function, under the DCR assumption.
    \item The function $g(\x)=\sqrt[3]{\x} \text{ mod } N$ has a ``trap-door'', in that there exists another way of computing it efficiently. For every $N$ (as specified above), there exists a \emph{key} $d\in\Z$ such that $g(\x)~=~\x^d~\text{ mod }~N$. Finding $d$ is efficient quantumly by using Shor's factoring algorithm \cite{shor99}, but hard classically under the DCR assumption.
\end{enumerate}

Observations (i) and (ii) can be leveraged to show that learning the functions $g_s$ from examples is also intractable. Indeed, Alexi \emph{et al.} \cite{alexi88} showed that a classical algorithm that could faithfully capture a single bit $g_s(\x)$ of the discrete cube root of $\x$, for even a $1/2 + 1/\text{poly}(n)$ fraction of all $\x\in\Z$, could also be used to reconstruct $g(\x),\ \forall \x\in\Z$, with high probability of success. Since, from observation (ii), the training data for the learning algorithm can also be generated efficiently classically from $N$, a classical learner that learns $g_s(\x)$ correctly for a $1/2 + 1/\text{poly}(n)$ fraction of all $\x\in\Z$ would then contradict the DCR assumption.

Observation (iii) allows us to define the following flipped model:
\begin{gather}\label{eq:cube_root-model}
f_{\thv}(\x)  = \Tr[\rho(\thv)O(\x)]\\
\rho(\thv)\!=\!\ket{d',s'}\!\bra{d',s'}\ \&\ O(\x) = \sum_{d',s'} \widehat{g}_{d',s'}(\x) \ket{d',s'}\!\bra{d',s'}.\nonumber
\end{gather}
That is, $\rho(\thv)$ (for $\thv=(N,s')$) specifies candidates for the key $d'$ and the parameter $s'$ of interest, while $O(\x)$ uses that information to compute
\begin{equation}
	\widehat{g}_{d',s'}(\x) = \begin{cases}
    1,&\text{if }\x^{d'} \text{ mod } N \in[s',s'+\frac{N-1}{2}],\\ 0, &\text{otherwise.}
    \end{cases}
\end{equation}
The state $\rho(\thv) = \ket{d,s'}\!\bra{d,s'}$ for the right key $d$ can be prepared efficiently using Shor's algorithm applied on $N$ (provided with the training data). As for $O(\x)$, it simply processes classically a bit-string to compute $\widehat{g}_{d',s'}(\x)$ efficiently, which corresponds to $g_s(\x)$ when $(d',s')=(d,s)$. Finding an $s'$ close to $s$ is an easy task given training data and $d'=d$. Since $\rho(\thv)$ is a computational basis state, this flipped model admits a trivial shadow model where a single computational basis measurement of $\rho(\thv)$ allows to evaluate $f_{\thv}(\x)$ classically for all $\x$. Therefore, we end up showing the following theorem:

\begin{thm}[Quantum advantage (informal)]\label{thm:quantum-advantage}
There exists a learning task where a shadow model first trained using a quantum computer then evaluated classically on new input data, can achieve an arbitrarily good learning performance, while any fully classical model cannot do significantly better than random guessing, under the hardness of classically computing the discrete cube root.
\end{thm}

In Appendix \ref{appdx:quantum-advantage} we formalize the statement of this result using the PAC framework and provide more details on the setting and the proofs.

\section{General shadow models}\label{sec:general-shadow}

As mentioned at the start of this paper, shadow models are not limited to shadowfied flipped models, and the main alternative proposals are based on the Fourier representation of quantum models \cite{schreiber22,landman22}. It is clear that Fourier models are defined very differently from flipped models, but one may wonder whether they nonetheless include shadowfied flipped models as a special case, or the other way around.

In this section, we first start by showing that there exist quantum models that admit shadow models (i.e., are shadowfiable) but cannot be shadowfied efficiently using a Fourier approach. This then motivates our proposal for a general definition of shadow models, and we show that, under this definition, all shadow models can be expressed as shadowfied flipped models. Finally, we show the existence of quantum models that are not shadowfiable at all under likely complexity theory assumptions.

\subsection{Shadow models beyond Fourier}\label{sec:beyond-Fourier}

An interesting approach to construct shadows of quantum models is based on their natural Fourier representation. It has been shown \cite{schuld21b,casas23} that quantum models can be expressed as generalized Fourier series of the form
\begin{equation}
    f_{\thv}(\x) = \sum_{\bm{\omega} \in \Omega} c_{\bm{\omega}}(\thv) e^{-i\bm{\omega}\cdot\x}
\end{equation}
where the accessible frequencies $\Omega$ only depend on properties of the encoding gates used by the model (notably the number of encoding gates and their eigenvalues). Since these frequencies can easily be read out from the circuit, one can proceed to form a shadow model by estimating their associated coefficients $c_{\bm{\omega}}(\thv)$ using queries of the quantum model $f_{\thv}(\x)$ at different values $\x$ and, e.g., a Fourier transform \cite{schreiber22}. Given a good approximation of these coefficients, one can then compute estimates of $f_{\thv}(\x)$ for arbitrary new inputs $\x$. We will refer to such a shadowing approach that considers the quantum model as a black-box, aside from the knowledge of its Fourier spectrum, as the Fourier shadowing approach.

Although we will be explicit about this in the next subsection, we will consider a shadowing procedure to be successful, if, with high probability, the resulting shadow model agrees with the original model on all inputs\footnote{We want the shadowing procedure to be successful independently of the data distribution under which the model should be trained, which justifies this definition. We discuss this point further in Appendix \ref{appdx:shadowfiability}.\label{note:shadowfiability}}, i.e.,
\begin{equation}
	\max_{\x\in\X}\big|f_{\thv}(\x) - \widetilde{f}_{\thv}(\x)\big| \leq \eps,
\end{equation}
for a specified $\eps\geq0$.

We show that the Fourier shadowing approach can suffer from an exponential sample complexity in the dimension of the input data $\x$, making it intractable for high-dimensional input spaces. To see this, consider the linear model:
\begin{equation}\label{eq:grover-model}
\begin{gathered}
f_{\y}(\x)  = \Tr[\rho(\x)O(\y)]\\
\rho(\x) = \bigotimes_{i=1}^n R_Y(x_i)\ket{0}\!\bra{0}R_Y^\dagger(x_i)\ \&\ O(\y) = \ket{\y}\!\bra{\y}.
\end{gathered}
\end{equation}
for $\x\in\R^n$ and $\y \in \{0,1\}^{\otimes n}$. Let us first restrict our attention to the domain $\x \in \{0,\pi\}^n$. It is quite clear that on this domain, $f_{\y}(\x) = \delta_{\x/\pi,\y}$ plays the role of a database search oracle, where the database has $2^n$ elements and a unique marked element $\y$. From lower bounds on database search, we know that $\Omega(2^n)$ calls to this oracle are needed to find $\y$ \cite{grover96}. This implies that a Fourier shadowing approach would require $\Omega(2^n)$ calls to $f_{\y}(\x) = \delta_{\x/\pi,\y}$ in order to guarantee $\max_{\x\in\X} \big|\widetilde{f}_{\thv}(\x)-f_{\thv}(\x)\big| \leq 1/4$. In Appendix \ref{appdx:beyond-Fourier}, we explain how this result can be generalized to the full domain $\x\in\R^n$, and we relate this bound on the sample complexity to the Fourier decomposition of the model.

On the other hand, note that the flipped model associated to $f_{\y}(\x)$ allows for a straightforward shadowing procedure. Indeed, by preparing $O(\y)$ and measuring it in the computational basis, one straightforwardly obtains $\y$ and can therefore classically compute the expectation value of any tensor product observable $\rho(\x)$ as specified by \cref{eq:grover-model}. Therefore, we have shown that there exist shadowfiable models that are not efficiently Fourier-shadowfiable, i.e., for which a shadowing procedure based solely on the knowledge of their Fourier spectrum and on black-box queries has query complexity that is exponential in the input dimension.

\subsection{All shadow models are shadows of flipped models} 
We give a general definition of shadow models that can encompass all methods that have been proposed to generate them. In contrast to the definition of classical surrogates proposed by Schreiber \emph{et al.}~\cite{schreiber22}, we give explicit definitions for the shadowing and evaluation phases of shadow models which makes explicit the need for a quantum computer in the shadowing phase. Indeed, as mentioned in the introduction, the term \emph{classical surrogate}\break has been used to describe both a classically evaluatable model obtained from a quantum shadowing procedure and a fully classical model trained directly on the data. We want to avoid this confusion in the definition of shadow models. We view a general shadowing phase as the generation of \emph{advice} that can be used to classically evaluate a quantum model. This advice is generated by the execution of quantum circuits that may or may not depend on the (trained) quantum circuit from the training phase. For instance, when we shadowfy a flipped model, we simply prepare the parametrized states $\rho(\thv)$ and use (randomized) measurements to generate an operationally meaningful classical description. In the case of Fourier shadowing, this advice is instead generated by evaluations of the quantum model $f_{\thv}(\x)$ for different inputs $\x\in\R^d$ that are rich enough to learn the Fourier coefficients of this model. We propose the following definition:

\begin{dfn}[General shadow model]\label{def:shadow-model}
Let $W_1(\thv),\allowbreak ...,\allowbreak W_M(\thv)$ be a sequence of $\O(\text{poly}(m))$-time quantum circuits applied on all-zero states $\ket{0}^{\otimes m}$, and that can potentially be chosen adaptively. Call $\omega(\thv) = (\omega_1(\thv), \ldots, \omega_M(\thv))$ the outcomes of measuring the output states of these circuits in the computational basis. A general shadow model is defined as:
\begin{equation}
f_{\thv}(\x) = \A(\x,\omega(\thv))
\end{equation}
where $\A$ is a classical $\O(\text{poly}(M,m,d))$-time algorithm that processes the outcomes $\omega(\thv)$ along with an input $\x\in\R^d$ to return the (real-valued) label $f_{\thv}(\x)$.
\end{dfn}

From this definition, a shadow model is a classically evaluatable model that uses quantum-generated advice. Crucially, this advice must be independent of the data points $\x$ we wish to evaluate the model on in the future. We distinguish the notion of a shadow model from that of a \emph{shadowfiable} quantum model, that is a quantum model that admits a shadow model:

\begin{dfn}[Shadowfiable model]\label{def:shadowfiable}
A model $f_{\thv}$ acting on $n$ qubits is said to be shadowfiable if, for $\eps,\delta > 0$, there exists a shadow model $\widetilde{f}_{\thv}$ such that, with probability $1-\delta$ over the quantum generation of the advice $\omega(\thv)$ (i.e., the shadowing phase), the shadow model satisfies\footref{note:shadowfiability}
\begin{equation}
	\max_{\x\in\X}\abs{f_{\thv}(\x) - \widetilde{f}_{\thv}(\x)} \leq \eps,
\end{equation}
and uses $m,M \in \O(\text{poly}(n,1/\eps,1/\delta))$ qubits and circuits to generate its advice $\omega(\thv)$.
\end{dfn}

While we have seen that there exist shadowfiable models that cannot be shadowfied efficiently using a Fourier approach, we show that all shadowfiable models as defined above can be approximated by shadowfiable flipped models.

\begin{lem}[Flipped models are shadow-universal]\label{lem:all-shadowfiable-flipped}
All shadowfiable models as defined in Defs.~\ref{def:shadow-model} and \ref{def:shadowfiable} can be approximated by flipped models $f_{\thv}(\x)~=~\Tr[\rho(\thv)O(\x)]$ with the guarantee that computational basis measurements of $\rho(\thv)$ and efficient classical post-processing can be used to evaluate $f_{\thv}(\x)$ to good precision with high probability.
\end{lem}

This result is essentially based on the observation that the evaluation of a general shadow model as defined in Def.~\ref{def:shadow-model} can be done entirely coherently. Instead of classically running the algorithm $\A$ using the random advice $\omega(\thv)$, one can quantumly simulate this algorithm (using a reversible execution) and execute it on the coherent advice $\rho(\thv) = \ket{\omega(\thv)}\!\bra{\omega(\thv)}$ generated by $\{W_1(\thv), ..., W_M(\thv)\}$ before the computational basis measurements. We refer to Appendix \ref{appdx:flipped-universal} for a more detailed statement and proof.

\subsection{Not all quantum models are shadowfiable}

\begin{figure}[t]
	\centering
	\includegraphics[width=\linewidth]{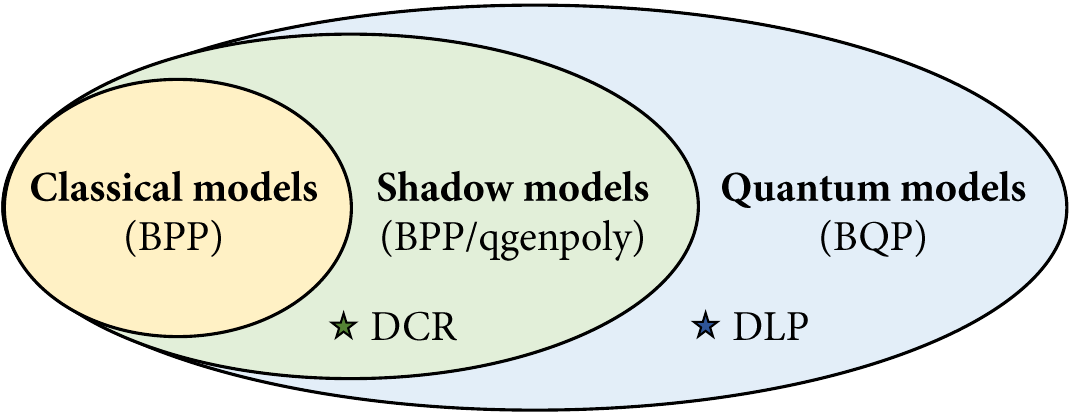}
	\caption{\textbf{Separations between classical, shadow, and quantum models.} Under the assumption that the discrete cube root (DCR) cannot be computed classically in polynomial time, we have a separation between shadow models (captured by the class \textsf{BPP/qgenpoly}) and classical models (in \textsf{BPP}). Under the assumption that there exist functions that can be computed in quantum polynomial time but not in classical polynomial time with the help of advice (i.e., \textsf{BQP}~$\not\subset$~\textsf{P/poly}), we have a separation between quantum models (universal for \textsf{BQP}) and shadow models (\textsf{BPP/qgenpoly}). A candidate function for this separation is the discrete logarithm (DLP).}
	\label{fig:models-venn}
\end{figure}

From the discrete cube root learning task, we already understand that a learning separation can be established between classical and shadowfiable models. We would also like to understand whether a learning separation exists between shadowfiable models and general quantum models, or equivalently, whether all quantum models are shadowfiable. We show that this also is not the case, under widely believed assumptions (see Fig.~\ref{fig:models-venn}).

\begin{thm}[Not all shadowfiable]\label{thm:not-all-shadowfiable}
    Under the assumption that \textsf{BQP} $\not\subset$ \textsf{P/poly}, there exist quantum models, i.e., models in \textsf{BQP}, that are not shadowfiable, i.e., that are not in \textsf{BPP/qgenpoly}.
\end{thm}

We start by noting that shadow models can be characterized by a complexity class we define as \textsf{BPP/qgenpoly},\footnote{Stands for Bounded-error Probabilistic Polynomial-time with quantumly generated (polynomial-time) advice of polynomial size. In this paper, we talk mostly about complexity classes for decision problems. Note however that for models, since these compute real-valued functions (represented to machine precision), we should instead consider the function-problem version of these complexity classes.} which contains all functions that can be computed efficiently classically with the help of polynomially-sized advice \emph{generated efficiently by a quantum computer}. This class is trivially contained in the standard class \textsf{BPP/poly}, which doesn't have any constraint on how the advice is generated and can be derandomized to \textsf{P/poly} (i.e., \textsf{BPP/poly}$=$\textsf{P/poly} \cite{adleman78}). Note however that \textsf{BPP/qgenpoly} constitutes a physically relevant class, since it only contains problems that can be solved efficiently by classical and quantum computers, as opposed to \textsf{P/poly}, which contains undecidable problems, such as a version of the halting problem. We refer to Appendix \ref{appdx:complexity-classes} for formal definitions of these complexity classes, and an in-depth discussion.

On the other hand, it is easy to show that quantum models (more precisely quantum linear models) can also represent any function in \textsf{BQP}, i.e., all functions that are efficiently computable on a quantum computer. For this, one simply takes a simple encoding of an $n$-bit input $\x$:
\begin{equation}
\rho(\x) = \bigotimes_{i=1}^n X_i^{x_i}\ket{0}\!\bra{0}X_i^{x_i}
\end{equation}
along with an observable
\begin{equation}
O_n = U_n^\dagger Z_1 U_n
\end{equation}
specified by an arbitrary $n$-qubit circuit $U_n$ in \textsf{BQP} and the Pauli-$Z$ operator applied on it first qubit. The resulting model $f_n(\x) = \Tr[\rho(\x)O_n]$ can then be used to decide any language in \textsf{BQP}.

Combining these two observations, we get that the proposition ``all quantum models are shadowfiable'' would imply that \textsf{BQP} $\subseteq$ \textsf{BPP/qgenpoly} $\subseteq$ \textsf{P/poly}, which violates the widely-believed conjecture \cite{aaronson11} that \textsf{BQP} $\not\subseteq$ \textsf{P/poly} (see Appendix \ref{appdx:BQP-P/poly} for a formal proof). To give an example of candidates of non-shadowfiable quantum models, the discrete logarithm $\log_g x \text{ mod } p$ (or even one bit of it) is provably in \textsf{BQP} but is not believed to be in \textsf{P/poly}. Therefore, a model that could be used to compute the discrete logarithm (e.g., the quantum model of Liu \emph{et al.}~\cite{liu20}) is likely not shadowfiable.

\section{Discussion}

In this work, we examined the class of quantumly trainable, classically evaluatable models we refer to as shadow models. Our analysis has shown that these models can be universally captured by a restricted family of quantum linear models, wherein data-encoding and variational operations are flipped compared to conventional quantum models. Furthermore, we demonstrated that shadow models belong to an intriguing complexity class, coined BPP/qgenpoly, exhibiting superiority over classical models (in BPP) but inferiority to fully quantum models (in BQP), based on prevalent complexity theory assumptions.

By presenting shadows models as flipped linear models, we illustrated how shadow tomography protocols could be applied straightforwardly to construct shadow models in practice. Yet, it is important to note a crucial distinction between a shadow tomography scenario and a shadow model: in the latter, one has control over the quantum state intended for shadowing. This distinction introduces new possibilities for devising `state-aware' shadow tomography protocols aimed at constructing shadow models. This could potentially alleviate some of the limitations of current classical shadow protocols. 

Considering our findings on learning separations, we identified a noteworthy characteristic of shadow models: their ability to quantumly compute useful advice for a classical evaluation algorithm, enabling them to tackle otherwise classically-intractable tasks. The example we presented, based on trap-door functions, readily allows for such constructions, but it remains somewhat contrived. Exploring similar constructions for physically-relevant problems, such as predicting ground state properties of complex quantum systems, would be an intriguing avenue for future research.

\vspace{5em}

\section*{Acknowledgments}
The authors would like to acknowledge Johannes Jakob Meyer, Franz Schreiber, and Richard Kueng for insightful discussions at different phases of this project. SJ acknowledges support from the Austrian Science Fund (FWF) through the projects DK-ALM:W1259-N27 and SFB BeyondC F7102. SJ also acknowledges the Austrian Academy of Sciences as a recipient of the DOC Fellowship. SJ thanks the BMWK (EniQma) and the Einstein Foundation (Einstein Research Unit on Quantum Devices) for their support. VD and CG acknowledge the support of the Dutch Research Council (NWO/ OCW), as part of the Quantum Software Consortium programme (project number 024.003.037). VD and RM acknowledge the support of the NWO through the NWO/NWA project Divide and Quantum. VD and SCM acknowledge the support by the project NEASQC funded from the European Union’s Horizon 2020 research and innovation programme (grant agreement No 951821). VD and SCM also acknowledge partial funding by an unrestricted gift from Google Quantum AI. This work was also in part supported by the Dutch National Growth Fund (NGF), as part of the Quantum Delta NL programme.

\section*{Data availability}
Data sharing is not applicable to this paper as no datasets were generated or analyzed during the current study.
 
\bibliography{references}

\clearpage

\appendix

\onecolumngrid

\section{Formal definitions}

\subsection{Linear models}

\begin{definition}[Conventional linear model]\label{def:conventional-model}
Let $U(\x)$ be an encoding quantum circuit that is parametrized by input data $\x\in\R^d$, $\rho_0$ a fixed input quantum state (diagonal in the computational basis), $V(\thv)$ a variational quantum circuit parametrized by a vector $\thv\in\R^p$ and $O = \sum_{i=1}^{m} w_i O_i$ an observable specified by a (trainable) linear combination of Hermitian matrices $\{O_i\}_{i=1}^m$. A conventional linear model is defined by the parametrized function:
\begin{equation}
f_{\thv}(\x)  = \Tr[\rho(\x)O(\thv)]
\end{equation}
for $\rho(\x) = U(\x)\rho_0U^\dagger(\x)$ and $O(\thv) = V^\dagger(\thv)OV(\thv)$ (when the weights $\{w_i\}_{i=1}^{m}$ are also trainable, we include them in the parameters $\thv$ of the model).
\end{definition}

\begin{definition}[Flipped model]\label{def:flipped-model}
Let $V(\thv)$ be a variational quantum circuit parametrized by a vector $\thv\in\R^p$, $\rho_0$ a fixed input quantum state (diagonal in the computational basis), $U(\x)$ an encoding quantum circuit that is parametrized by input data $\x\in\R^d$ and $O_{\x} = \sum_{i=1}^{m} w(\x)_i O_i$ an observable specified by a linear combination of Hermitian matrices $\{O_i\}_{i=1}^m$, weighted by a data-dependent function $w: \R^d \rightarrow \R^m$. A flipped model is defined by the parametrized function:
\begin{equation}
f_{\thv}(\x)  = \Tr[\rho(\thv)O(\x)]
\end{equation}
for $\rho(\thv) = V(\thv)\rho_0V^\dagger(\thv)$ and $O(\x) = U^\dagger(\x)O_{\x}U(\x)$.
\end{definition}

\subsection{Shadow models}\label{appdx:shadow-models}

\begin{definition}[Shadow model]\label{def:shadow-model-formal}
Let $\{W_1(\thv), ..., W_M(\thv)\}$ be a sequence of $m$-qubit unitary circuits that are dependent on a parameter vector $\thv\in\R^p$, and can potentially be chosen adaptively. We define the \emph{quantum-generated} advice $\omega(\thv) = (\omega_1(\thv), \ldots, \omega_M(\thv))$ as the measurement outcomes $\omega_i(\thv)$ obtained by measuring the states $W_i(\thv)\ket{0}^{\otimes m}$ in the computational basis (and a description of their associated circuits $W_i(\thv)$). A shadow model is defined as the parametrized function:
\begin{equation}
f_{\thv}(\x)  = \A(\x,\omega(\thv))
\end{equation}
for $\A$ a classical $\O(\text{poly}(m,M,d))$-time algorithm that takes as input the advice $\omega(\thv)$, a data vector $\x\in\R^d$ and outputs a real-valued label $f_{\thv}(\x)$
\end{definition}

Examples of shadow models:
\begin{itemize}
	\item Take a flipped model $f_{\thv}(\x)  = \Tr[\rho(\thv)O(\x)]$ for $\rho(\thv)$ a quantum state generated by a circuit $V(\thv)$ applied to $\ket{0}^{\otimes n}$ and $O(\x) = \sum_{i=1}^{m} w(\x)_i P_i$ where $\{P_i\}_{i=1}^{m}$ are all $k$-local Pauli strings acting on $n$ qubits.\\
	A simple shadow model associated to this flipped model consists in estimating all the expectation values $\expval{P_i}\approx\Tr[\rho(\thv)P_i]$ via repeated measurements of $\rho(\thv)$ in the eigenbasis of each of the $m = {n\choose k}3^k$ Pauli strings $P_i$, and taking their weighted combination $f_{\thv}(\x) = \sum_{i=1}^{m} w(\x)_i \expval{P_i}$. In this case, the unitary circuits $\{W_1(\thv), ..., W_M(\thv)\}$ are simply obtained by $V(\thv)$ followed by a basis change unitary (corresponding to the Pauli basis of $P_i$). As for the classical algorithm $\A$, this is simply a collection of mean estimators that compute estimates $\expval{P_i}$ out of measurement outcomes, followed by the computation of a weighted sum. The number of measurements needed for this shadow model to guarantee $\abs{\widetilde{f}_{\thv}(\x)-f_{\thv}(\x)} \leq \eps,\ \forall \x\in\R^d$ is $M \in \widetilde{\O}\left(\frac{m\max_{\x}\norm{w(\x)}_1^2}{\eps^2}\right)$. Indeed, estimating each $\expval{P_i}$ to additive error $\frac{\eps}{\max_{\x}\norm{w(\x)}_1}$ allows us to guarantee the desired total additive error, and each of this estimates can be obtained using $\widetilde{\O}\left(\frac{\max_{\x}\norm{w(\x)}_1^2}{\eps^2}\right)$ samples.\\
	A more interesting shadow model relies on Pauli classical shadows \cite{huang20} where random Pauli measurements are used to construct $\omega(\thv)$. Median-of-mean estimators then use these measurement outcomes to compute empirical estimates $\hat{\rho}(\thv)$ of $\rho(\thv)$ and approximate all expectation values $\Tr[\rho(\thv)P_i]$. The advantage of this shadow model is that it requires $M \in \widetilde{\O}\left(\frac{3^k\max_{\x}\norm{w(\x)}_1^2}{\eps^2}\right)$ measurements for the same guarantees as the shadow model above, which constitutes savings of a factor $\widetilde{\O}\left({n\choose k}\right)$.
	\item Another interesting flipped model that can be turned into a shadow model: $f_{\thv}(\x)  = \Tr[\rho(\thv)O(\x)]$ for $\rho(\thv)$ arbitrarily defined and $O(\x) = \ket{\psi(\x)}\!\bra{\psi(\x)}$, for some pure states $\ket{\psi(\x)}$. Given that, $\Tr[O(\x)^2] = 1,\ \forall \x \in \R^d$, then Clifford classical shadows \cite{huang20} allow to construct a representation $\omega(\thv)$ of $\rho(\thv)$ that guarantees $\abs{\widetilde{f}_{\thv}(\x)-f_{\thv}(\x)} \leq \eps,\ \forall \x\in\R^d$ using only $M \in \widetilde{\O}\left(\frac{1}{\eps^2}\right)$ measurements. However, for this estimation to be computationally efficient, the states $\ket{\psi(\x)}$ need to be stabilizer states, or be generated by few (i.e., $\O(\log(n))$) non-Clifford gates \cite{leone23}.
	\item Consider the model $f_{\thv}(\x)  = \sum_{i=1}^{m} w(\x)_i \Tr[\rho(\thv)P_i]^2$ for $m = 4^n$, i.e., $\{P_i\}_{i=1}^{m}$ are all $n$-qubit Pauli strings. By preparing two-copy states $\rho(\thv)\otimes\rho(\thv)$ and performing simultaneous Bell measurements between pairs of qubits of these two copies, $M = \O\left(\frac{1}{\eps^2}\right)$ such measurements give a $\omega(\thv)$ rich enough to compute any $\Tr[\rho(\thv)P_i]^2$ to precision $\eps$ \cite{huang21}. Therefore, evaluating $f_{\thv}(\x)$ requires only $M \in \widetilde{\O}\left(\frac{\max_{\x}\norm{w(\x)}_1^2}{\eps^2}\right)$ measurements using this shadow model, compared to $2^{\Omega(n)}$ for a shadow model that would construct $\omega(\thv)$ using single-copy measurements only (assuming that for all $i \in \{1, \ldots, m\}$, there exists an $\x \in \R^d$ such that $w(\x)_i \neq 0$). For a $w(\x)$ that is $k$-sparse for all $\x$, with $k \ll m$, this constitutes an exponential separation.
\end{itemize}

\subsection{Complexity classes\label{appdx:complexity-classes}}

\begin{definition}[\textsf{BQP}]\label{def:BQP}
    A language $L$ is in \textsf{BQP} if and only if there exists a polynomial-time uniform family of quantum circuits $\{U_n: n \in \mathbb{N}\}$, such that
\begin{enumerate}
    \item For all $n \in \mathbb{N}$, $U_n$ takes as input an $n$-qubit computational basis state, and outputs one bit obtained by measuring the first qubit in the computational basis.
    \item For all $x \in L$, the probability that the output of $U_{|x|}$ applied on the input $x$ is $1$ is greater or equal to $2/3$.
    \item For all $x \notin L$, the probability that the output of $U_{|x|}$ applied on the input $x$ is $0$ is greater or equal to $2/3$.
\end{enumerate}
\end{definition}

\begin{definition}[\textsf{P/poly}]\label{def:P/poly}
    A language $L$ is in \textsf{P/poly} if and only if there exists a polynomial-time classical algorithm $\A$ and a sequence of polynomial-size advice strings $\{\alpha_n\in\{0,1\}^{\text{poly}(n)}\}_{n \in \mathbb{N}}$ such that for all $n \in \mathbb{N}$ and all $x\in\{0,1\}^n$:
    \begin{equation}
    \A(x,\alpha_n) = 1 \iff x \in L.
    \end{equation}
\end{definition}

\begin{definition}[\textsf{BPP/qgenpoly}]\label{def:BPP/qgenpoly}
    A language $L$ is in \textsf{BPP/qgenpoly} if and only if there exists a polynomial-time uniform family of quantum circuits $\{U_n: n \in \mathbb{N}\}$ and a polynomial-time probabilistic classical algorithm $\A$, such that
\begin{enumerate}
    \item For all $n \in \mathbb{N}$, $U_n$ takes as input the computational basis state $\ket{0}^{\otimes m}$ for $m\in\O(\text{poly}(n))$, and outputs $m$ bits obtained by measuring all qubits in the computational basis. This constitutes the \emph{quantum-generated} advice $\omega_n$.
    \item For all $x \in \{0,1\}^*$, $\A$ takes as input $x$ and $\omega_{|x|}$.
    \item For all $x \in L$, the probability that $\A(x,\omega_{|x|})$ outputs $1$ is greater or equal to $2/3$, taken over the randomness of $\omega_{|x|}$ and the internal randomness of $\A$.
    \item For all $x \notin L$, the probability that $\A(x,\omega_{|x|})$ outputs $0$ is greater or equal to $2/3$, taken over the randomness of $\omega_{|x|}$ and the internal randomness of $\A$.
\end{enumerate}
\end{definition}

The following inclusions are easy to show:
\begin{equation}
\textsf{BPP} \stackrel{(i)}{\subseteq} \textsf{BPP/qgenpoly} \stackrel{(ii)}{\subseteq} \textsf{BQP}.
\end{equation}
(\emph{i}) follows from making the quantum-generated advice $\omega_n$ empty in the definition of \textsf{BPP/qgenpoly}. (\emph{ii}) follows from the ability to efficiently simulate classical computations on a quantum computer. As illustrated in \cref{fig:shadow-flipped}, the algorithm $\A$ in the definition of \textsf{BPP/qgenpoly} can be simulated unitarily, and absorbed in the uniform family of quantum circuits $\{U_n: n \in \mathbb{N}\}$, resulting in polynomial-time quantum circuits that fit the definition of \textsf{BQP}.

Our learning separations results, i.e., Theorem \ref{thm:quantum-advantage} (Lemmas \ref{lem:harndess-DCR} and \ref{lem:learnability-DCR} in the Appendix) and Theorem \ref{thm:not-all-shadowfiable} (Lemma \ref{thm:not-all-shadowfiable-extended} in the Appendix), can be seen as evidence that the inclusions (\emph{i}) and (\emph{ii}) are strict, based on complexity-theory assumptions. Other notable inclusions that are useful to prove our results are:
\begin{equation}
\textsf{BPP/qgenpoly} \subsetneq \textsf{BPP/poly} = \textsf{P/poly}.
\end{equation}
The equality $\textsf{BPP/poly} = \textsf{P/poly}$ follows from the derandomization results of Adleman \cite{adleman78}, which show that the random errors made by an algorithm in \textsf{BPP/poly} can be canceled by an appropriate choice of the random bits used by the randomized algorithm, which are then appended to the poly-sized advice to obtain an algorithm in \textsf{P/poly}. The inclusion $\textsf{BPP/qgenpoly} \subseteq \textsf{BPP/poly}$ simply comes from the fact that \textsf{BPP/poly} is not restricted in \emph{how} the poly-sized advice is generated, and the remainder of its definition is identical to that of \textsf{BPP/qgenpoly}. The restriction we make on \emph{how} the advice is generated in \textsf{BPP/qgenpoly} makes it a physically relevant complexity class, as opposed to \textsf{(BP)P/poly}. All problems in \textsf{BPP/qgenpoly} can be solved efficiently (i.e., in polynomial time) using (classical and) quantum computers, while \textsf{P/poly} notably contains undecidable problems. Consider for instance the unary version of the halting problem: $\textrm{UHALT} = \{1^n, \textrm{where }n \textrm{ encodes }(M,x)\textrm{ such that the Turing machine }M\textrm{ halts on }x\}$ is an undecidable language (as any algorithm that would decide it would also be able to decide the traditional halting problem), but by considering the advice $\alpha_n = 1$ if $1^n \in \textrm{UHALT}$ and $0$ otherwise (uniquely defined for each input size $n$), one trivially obtains an algorithm in \textsf{P/poly} that solves it. The fact that this advice cannot be generated by a uniform family of (poly-time) circuits is irrelevant for the class \textsf{P/poly}, which is at the source of this result. However it is relevant for the class \textsf{BPP/qgenpoly}, and, in fact, having $\textrm{UHALT} \in \textsf{BPP/qgenpoly}$ would mean that one could solve an undecidable problem using uniform (poly-time) circuits. The impossibility of this hypothetical result gives $\textsf{BPP/qgenpoly} \subsetneq \textsf{P/poly}$.

\section{Properties of flipped models}

\subsection{Sample complexity of evaluating quantum models}\label{appdx:sample-complexity}

Consider a linear quantum model (either conventional or flipped) of the form 
\begin{equation}
    f_{\y}(\x) = \Tr[\rho(\x)O(\y)],
\end{equation}
for a quantum state $\rho(\x)$ parametrized by a vector $\x\in\R^d$ and $O(\y)$ a Hermitian observable parametrized by a vector $\y\in\R^p$. Assume that we can prepare single copies of $\rho(\x)$ and that we can measure them in the eigenbasis of $O(\y)$. We ask: given error parameters $\eps,\delta >0$, how many such measurements of $\rho(\x)$ do we need in order to compute an estimate $\widehat{f}_{\y}(\x)$ of $f_{\y}(\x)$ such that $\abs{\widehat{f}_{\y}(\x)-f_{\y}(\x)}\leq \eps$ with success probability at least $1-\delta$.

It is easy to see that this problem corresponds to a simple Monte Carlo mean estimation. Indeed, we can write a decomposition of $O(\y)$ in its eigenbasis as:
\begin{equation}
    O(\y) = \sum_{i} \lambda_i(\y) \ket{\phi_i}\!\bra{\phi_i}
\end{equation}
where $\lambda_i(\y)$ is a real eigenvalue (since $O(\y)$ is Hermitian) associated to the eigenstate $\ket{\phi_i}\!\bra{\phi_i}$ (these eigenstates can in general also depend on $\y$, but we do not write this dependence explicitly for ease of notation). We can also write a decomposition of $\rho(\x)$ in this same basis as :
\begin{equation}
    \rho(\x) = \sum_{i,j} \rho_{i,j}(\x) \ket{\phi_i}\!\bra{\phi_j}
\end{equation}
such that $\Tr[\rho(\x)] = \sum_{i}\rho_{i,i}(\x) = 1$ by the unit-trace property of $\rho(\x)$, and $\bra{\phi_i}\rho(\x)\ket{\phi_i} = \rho_{i,i}(\x) \geq 0$ from its positive semi-definiteness. From these two properties, we deduce that $\{\rho_{i,i}(\x)\}_{i}$ defines a probability distribution over the eigenstates $\{\ket{\phi_i}\!\bra{\phi_i}\}_{i}$. Therefore, we can see that:
\begin{align}
    \Tr[\rho(\x)O(\y)] &= \Tr[\sum_{i,j,k} \rho_{i,j}(\x)\lambda_k(\y) \ket{\phi_i}\!\bra{\phi_j}\ket{\phi_k}\!\bra{\phi_k}]\\
    &= \Tr[\sum_{i,j} \rho_{i,j}(\x)\lambda_j(\y)\ket{\phi_i}\!\bra{\phi_j}]\\
    &= \sum_{i} \rho_{i,i}(\x)\lambda_i(\y)
\end{align}
simply corresponds to the expectation value of the random variable $i \mapsto \lambda_i(\y)$ under the probability distribution $\{\rho_{i,i}(\x)\}_{i}$, i.e., the probability distribution obtained by measuring $\rho(\x)$ in the eigenbasis of $O(\y)$.

Therefore, we can use known results from (classical) Monte Carlo estimation to bound the sample complexity of evaluating this mean value, and therefore the quantum model. With the assumption that $\rho(\x)$ is given as a black-box and that it can generate arbitrary quantum states (and therefore arbitrary distributions $\{\rho_{i,i}(\x)\}_{i}$), the only property of the random variable $i \mapsto \lambda_i(\y)$ we can use to bound the sample complexity is its bounded domain. Indeed, without additional assumptions of the distribution, we have a tight sample complexity bound of \begin{equation}
    \Theta\left(\frac{B^2\log(\delta^{-1})}{\eps^2}\right)
\end{equation}
samples in order to estimate the mean of a random variable taking values in $[-B,B]$, to precision $\eps$ and with probability of success $1-\delta$ \cite{dagum00,canetti95}. In the case of a quantum model, the random $i \mapsto \lambda_i(\y)$ takes values in $[-\norm{O(\y)}_\infty,\norm{O(\y)}_\infty]$, where $\norm{O(\y)}_\infty$ is the spectral norm of the observable $O(\y)$. Therefore, the sample complexity of estimating a quantum model $f_{\y}(\x) = \Tr[\rho(\x)O(\y)]$ is in
\begin{equation}
    \Theta\left(\frac{\norm{O(\y)}_\infty^2\log(\delta^{-1})}{\eps^2}\right),
\end{equation}
in the absence of any constraint (or information) on the quantum states $\rho(\x)$. 

\subsection{Generalization performance}

In this section, we study the generalization performance of flipped models. Our result can be stated informally in the following lemma:

\begin{lemma}[Generalization bounds (informal)]\label{lem:gen-bounds}
Consider a flipped model $f_{\thv}$ that acts on $n$ qubits and has a bounded observable norm $\norm{O}_\infty$. If this model achieves a small training error $\abs{f_{\thv}(\x) -f(\x)}\leq\eta$ for all $\x$ in a dataset of size $M$, then it also has a small expected error $\abs{f_{\thv}(\x) -f(\x)}\leq2\eta$ with probability $1-\eps$ over the data distribution, provided that the size of the dataset scales as $M \geq \widetilde{\Omega}\left(\frac{n\norm{O}_\infty^2}{\eps\eta^2}\right)$.
\end{lemma}

To prove this result, we take an approach very similar to that of Aaronson \cite{aaronson07}, where we lower bound the number of qubits $n$ and spectral norm $\norm{O}_\infty$ needed by a flipped model to encode arbitrary $k$-bit strings, in a way that can be recovered efficiently via repeated measurements. These bounds naturally allow us to upper bound the fat-shattering dimension of flipped models, a complexity measure that is widely used in generalization bounds \cite{anthony95}.

\subsubsection{Encoding bit-strings in flipped models}

\begin{theorem}\label{thm:lower-bound-shattering}
Let $k$ and $n$ be positive integers with $k>n$. For all $k$-bit strings $\y = y_1 \ldots y_k$, let $\rho(\y)$ be an $n$-qubit mixed state that ``encodes'' $\y$, meaning each bitstring $\y$ is associated to an arbitrary $n$-qubit quantum state $\rho(\y)$. Suppose there exist Hermitian observables $O_1, \ldots, O_k$ with spectral norms $\norm{O_i}_\infty$ such that we call $\norm{O}_\infty = \max_i \norm{O_i}_\infty$, as well as real numbers $\alpha_1, \ldots, \alpha_k$, such that, for all $\y\in\{0,1\}^k$ and $i \in \{1,\ldots,k\}$,\\[-0.5em]

\hspace{1em} (i) if $y_i=0$, then $\Tr[\rho(\y)O_i] \leq \alpha_i - \gamma$, and\\[-0.5em]

\hspace{1em} (ii) if $y_i=1$, then $\Tr[\rho(\y)O_i] \geq \alpha_i + \gamma$.\\[-0.5em]

Then $n\norm{O}_\infty^2/\gamma^2 \in \Omega(k)$.
\end{theorem}
\begin{proof}
We take a similar approach to the proof of Aaronson \cite{aaronson07}, in which we show that a combination of an encoding $\rho(\y)$ and observables $O_1, \ldots, O_k$ that satisfies guarantees (i) and (ii) would need $n\norm{O}_\infty^2/\gamma^2$ to scale linearly in the length $k$ of the bit-strings it encodes in order not to contradict with Holevo's bound.\\

Suppose by contradiction that such an encoding scheme exists with $n\norm{O}_\infty^2/\gamma^2 \in o(k)$. We first adapt to the setting of Aaronson by constructing two-outcome POVMs $\{E_i,I-E_i\}$ out of the observables $O_i$. That is, we take the general Hermitian matrices $O_i$ with eigenvalues in $[\lambda_{\text{min}}, \lambda_{\text{max}}] \subset [-\norm{O}_\infty, \norm{O}_\infty]$, and transform them into Hermitian matrices $E_i$ with eigenvalues in $[0, 1]$, such that the POVM $\{E_i,I-E_i\}$ accepts $\rho$ (i.e., outputs $1$) with probability $\text{Tr}(\rho E_i)$, and rejects $\rho$ (i.e., outputs $0$) with probability $1-\text{Tr}(\rho E_i)$. Specifically, we define
\begin{equation}
E_i = \frac{O_i + \abs{\lambda_{\text{min}}}I}{\abs{\lambda_{\text{min}}} + \abs{\lambda_{\text{max}}}}.
\end{equation}
Conditions $(i)$ and $(ii)$ then translate to:
\begin{equation}
    \begin{cases}
    (i')\text{ if } y_i=0 \text{, then } \Tr[\rho(\y)E_i] \leq  \frac{\alpha_i + \abs{\lambda_{\text{min}}}}{\abs{\lambda_{\text{min}}} + \abs{\lambda_{\text{max}}}} - \frac{\gamma}{\abs{\lambda_{\text{min}}} + \abs{\lambda_{\text{max}}}}\\
    (ii')\text{ if } y_i=1 \text{, then } \Tr[\rho(\y)E_i] \geq  \frac{\alpha_i + \abs{\lambda_{\text{min}}}}{\abs{\lambda_{\text{min}}} + \abs{\lambda_{\text{max}}}} + \frac{\gamma}{\abs{\lambda_{\text{min}}} + \abs{\lambda_{\text{max}}}}
    \end{cases}
\end{equation}
From here, by noting that $\abs{\lambda_{\text{min}}} + \abs{\lambda_{\text{max}}} \leq 2 \norm{O}_\infty$, we can directly apply Theorem 2.6 of Aaronson \cite{aaronson07} and get our result, but we detail the reasoning further for clarity.

We first need to amplify the probability that we correctly identify whether $y_i=0$ or $1$ from measuring copies of $\rho(\y)$, since the probabilities obtained from $E_i$ can be arbitrarily small. Consider an amplified scheme, where each bit-string $\y \in \{0,1\}^k$ is encoded by the tensor product $\rho(\y)^{\otimes \ell}$, for some $\ell>1$ to be defined later. For all $i \in \{1, \ldots, k\}$, let $\{E_i^*,I-E_i^*\}$ be the amplified POVM that applies $\{E_i,I-E_i\}$ to each of the $\ell$ copies of $\rho(\y)$ and accepts if and only if at least $\widetilde{\alpha}_i\ell = \frac{\alpha_i + \abs{\lambda_{\text{min}}}}{\abs{\lambda_{\text{min}}} + \abs{\lambda_{\text{max}}}}\ell$ of these POVMs do. For all $j \in \{1,\ldots, \ell\}$, call $X_i^{(j)}$ the random variable that takes the value $1$ if $\{E_i,I-E_i\}$ accepts the $j$-th copy of $\rho(\y)$ (i.e., with probability $p_i = \text{Tr}(\rho(\y) E_i)$), and value $0$ otherwise.\\

Consider the case where $y_i = 0$. We have $\text{Tr}[O_i \rho(y)] \leq \alpha_i - \gamma$, which implies that $\widetilde{\alpha}_i \geq p_i + \frac{\gamma}{\abs{\lambda_{\text{min}}} + \abs{\lambda_{\text{max}}}}$. Therefore, the probability that at least $\widetilde{\alpha}_i\ell$ of the POVMs accept is then:
\begin{equation}
	P(\overline{X}_i \geq \widetilde{\alpha}_i) \leq P\left(\overline{X}_i \geq p_i + \frac{\gamma}{\abs{\lambda_{\text{min}}} + \abs{\lambda_{\text{max}}}}\right)
\end{equation}
for  $\overline{X}_i = \frac{1}{\ell}\sum_{j=1}^{\ell} X_i^{(j)}$. From the Chernoff bound, we hence get:
\begin{equation}
	P(\overline{X}_i \geq \widetilde{\alpha}_i) \leq e^{-2\left(\frac{\gamma}{\abs{\lambda_{\text{min}}} + \abs{\lambda_{\text{max}}}}\right)^2l}.
\end{equation}
To guarantee an acceptance probability $\text{Tr}(\rho(\y)^{\otimes \ell}E_i^*) \leq 1/3$, it is then sufficient to take $\ell = \left\lceil\frac{2\log(3)\norm{O}_\infty^2}{\gamma^2}\right\rceil$. A similar analysis holds for the case $y_i = 1$.\\

From here, the result we use that derives from Holevo's bound is Theorem 5.1 of Ambainis \emph{et al.}~\cite{ambainis02}. It states that in order for the POVMs $\{E_i^*,I-E_i^*\}$ to correctly identify whether $y_i=0$ or $1$ with probability of failure less than $1/3$, we need a number of qubits $n\ell \geq (1-H(1/3))k$, where $H$ is the binary entropy function. This implies that $n\norm{O}_\infty^2/\gamma^2 \geq \frac{2(1-H(1/3))}{\log(3)} k \in \Omega(k)$. 
\end{proof}

\subsubsection{Generalization bounds of flipped models}

The conditions $(i)$ and $(ii)$ of \cref{thm:lower-bound-shattering} are very similar to that of a fat-shattering dimension of a concept class.

\begin{definition}
    Let $\X$ be a data space, let $\C$ be a class of functions from $\X$ to $\R$, and let $\gamma>0$. The fat-shattering dimension of the concept class $\C$ at width $\gamma$, denoted $\textnormal{fat}_{\C}(\gamma)$, is defined as the size $k$ of the largest set of points $\{\x^{(1)}, \ldots, \x^{(k)}\}$ for which there exist real numbers $\alpha_1, \ldots, \alpha_k$ such that for all $\y\in\{0,1\}^k$, there exists a $f\in\C$ that satisfies, for all $i \in \{1,\ldots,k\}$,\\[-0.5em]
    
    \hspace{1em} (i) if $y_i=0$, then $f(\x) \leq \alpha_i - \gamma$, and\\[-0.5em]
    
    \hspace{1em} (ii) if $y_i=1$, then $f(\x) \geq \alpha_i + \gamma$.\\[-0.5em]
\end{definition}

By comparing this definition with the statement of \cref{thm:lower-bound-shattering}, we can show:

\begin{corollary}\label{cor:fat-shattering-flipped}
Consider a flipped model $f_{\thv}(\x) = \Tr[\rho(\thv) O(\x)]$ defined on a data space $\X$, using $n$-qubit quantum states and observables with spectral norm $\norm{O}_\infty=\sup_{\x\in\X} \norm{O(\x)}_\infty$. Call $\C_{n,O} = \{f_{\thv}\}_{\thv}$ the concept (or hypothesis) class associated to this model. Then, for all $\gamma > 0$, we have $\textnormal{fat}_{\C_{n,O}}(\gamma) \in \O (n\norm{O}_\infty^2/\gamma^2)$.
\end{corollary}
\begin{proof}
We note that since, for all $\x\in\X$, $O(\x)$ lives in a manifold of all observables with spectral norm $\norm{O}_\infty$, then the fat-shattering dimension of $\C_{n,O}$ is upper bounded by that of the concept class $\widetilde{\C}_{n,O} = \{ O' \mapsto \text{Tr}[\rho({\thv}) O']\}_{\thv}$ defined on the input space of observables $O'$ that satisfy $\norm{O'}_\infty \leq \norm{O}_\infty$. Theorem \ref{thm:lower-bound-shattering} immediately yields an upper bound for $\textnormal{fat}_{\widetilde{\C}_{n,O}}(\gamma)$, when we identify $\rho(\thv)$ in this corollary to $\rho(\y)$ in the Theorem, and observe that the number of observables $O_1, \ldots, O_k$ that can be $\gamma$-shattered (i.e., conditions $(i)$ and $(ii)$ for all labelings) must satisfy $k \in \O (n\norm{O}_\infty^2/\gamma^2)$.
\end{proof}

To obtain generalization bounds on the performance of flipped models, we combine this bound on their fat-shattering dimension with standard results from learning theory, e.g.:

\begin{theorem}[Anthony and Barlett \cite{anthony95}]
Let $\X$ be a data space, let $\C$ be a class of functions from $\X$ to $\R$, and let $\D$ be a probability measure over $\X$. Fix an element $f \in \R^\X$, as well as error parameters $\eps, \eta, \gamma, \delta >0$ with $\gamma > \eta$. Suppose that we draw $m$ samples $X = (\x^{(1)}, \ldots, \x^{(m)})$ from $\X$ according to $\D$, and choose any hypothesis $h \in \C$ such that $\abs{h(\x) -  f(\x)} \leq \eta$ for all $\x \in X$. Then, there exists a positive constant $K$ such that, provided
$$ m \geq \frac{K}{\eps}\left( \textnormal{fat}_{\C} \left( \frac{\gamma-\eta}{8} \right) \log^2\left( \frac{\textnormal{fat}_{\C} \left( \frac{\gamma-\eta}{8} \right)} {(\gamma-\eta)\eps} \right) + \log\left(\frac{1}{\delta}\right)\right),$$
with probability $1-\delta$,
$$\textnormal{Pr}_{\x \in \D}[\abs{h(\x)-f(\x)}>\gamma] \leq \eps.$$
\end{theorem}

From Corollary \ref{cor:fat-shattering-flipped}, we have that $$m \in \widetilde{\Omega} \left(\frac{1}{\eps} \left(\frac{n\norm{O}_{\infty}^2}{(\gamma-\eta)^2}  + \log(\frac{1}{\delta})\right) \right)$$ samples suffice to get these generalization guarantees. This proves Lemma \ref{lem:gen-bounds}.

\subsection{Flipping bounds}\label{appdx:flipping-bounds}

In this section, we study mappings between conventional and flipped models (most importantly from conventional to flipped, but our flipping bounds can be used either way). We find that the important quantity that governs the trade-off in resources between these models is the observable trace norm $\norm{O}_1 = \Tr[\sqrt{O^2}]$ of the model to be mapped. Since all observables $O(\y)$ (where $\y$ is either a data vector or a parameters vector, depending on the model) have to be turned into unit-trace density matrices, their eigenvalues need (i) either to be normalized when preserving the number of qubits of the model, or (ii) be encoded in more qubits than in the original model (e.g., by using a binary encoding of all the eigenvalues of $O(\y)$). Each of these options has its disadvantages: (i) normalizing eigenvalues introduces an overhead in the spectral norm $\norm{O'}_{\infty}$ of the observable of the resulting model, which results in an overhead in the number of measurements needed to evaluate this model to the same precision as the original one. As for (ii), we show that the number of qubits of the new model would need to scale quadratically with $\norm{O}_1$ in general, which is commonly an exponential quantity in the number of qubits of the original model (e.g., when $O$ is a Pauli). We show the following lemma:

\begin{lemma}[Flipping bounds (informal)]\label{lem:flipping-bounds}
Any conventional linear model $f_{\thv}(\x)  = \Tr[\rho(\x)O(\thv)]$ acting on $n$ qubits and with a bounded observable trace norm $\norm{O}_1 \leq d$ admits an equivalent flipped model $\widetilde{f}_{\thv}(\x)  = \Tr[\rho'(\thv)O'(\x)]$ acting on $m=n+1$ qubits and with observable spectral norm $\norm{O'}_\infty = d$. The bound on the spectral norm is essentially tight in the regime where $n,m\!\in\!\O(\log(d))$, i.e., in this case we have $\norm{O'}_\infty\!\geq\!\widetilde{\Omega}(d)$.
\end{lemma}

	\subsubsection{Upper bounds}
	
\begin{theorem}\label{thm:upper-bound-flipping}
Given a specification of a conventional quantum model $f_{\thv}(\x) = \Tr[\rho(\x) O(\thv)]$ acting on $n$ qubits, with a known (upper bound on the) trace norm $\norm{O}_1$ of its observable, one can construct an equivalent flipped model $\widetilde{f}_{\thv}(\x) = \Tr[\rho'(\thv) O'(\x)]$ acting on $n+ 1$ qubits such that $\widetilde{f}_{\thv}(\x) = f_{\thv}(\x),\ \forall \x,\thv$ and $\norm{O'}_\infty = \norm{O}_1$.
\end{theorem}
\begin{proof}
From Def.~\ref{def:conventional-model}, we assume that, in the definition of the conventional model, $\rho(\x)$ is obtained by applying a unitary $U(\x)$ on a known quantum state $\rho_0$ that is diagonal in the computational basis (i.e., is a mixture of computational basis states). As for $O(\thv)$, we assume that it is specified by a unitary $V(\thv)$ and a weighted sum of $d$ Hermitian operators $O_i$, such that $O(\thv) = \sum_{i=1}^d w_i V^\dagger(\thv) O_i V(\thv)$ and for each $i\in\{1, \ldots, d\}$ we know how to decompose $O_i$ as $O_i = \sum_{j=0}^{2^n-1} \lambda_{i,j} W_i^\dagger\ket{j}\!\bra{j}W_i$, for some known $\lambda_{i,j}$'s and $W_i$'s. Note that, as opposed to the parameters that specify $V(\thv)$, the weights $w_i$ influence the trace norm $\norm{O(\thv)}_1$. Therefore we need to pay attention to the fact that $\norm{O(\thv)}_1$ is only upper bounded by $\norm{O}_1 = \sup_{\thv}\norm{O(\thv)}_1$, and that these two quantities are not always equal.\\

\begin{figure}[ht]
  \centering
\begin{minipage}{.82\linewidth}
\DecMargin{0.2em}
\begin{algorithm2e}[H]\label{alg:flipping}
\caption{Flipped evaluation of a conventional model}
 \KwIn{an $n$-qubit unitary $U(\x)$ and a quantum state $\rho_0$ (diagonal in the computation basis) such that $\rho(\x)=U(\x)\rho_0U^\dagger(\x)$, $n$-qubit unitaries $V(\thv)$ and $\{W_i\}_{1\leq i\leq d}$, real values $\{w_i\}_{i=1}^d$, $\{\lambda_{i,j}\}^{1\leq i\leq d}_{0\leq j\leq 2^n-1}$, such that $O(\thv) = \sum_i w_i V^\dagger(\thv) O_i V(\thv)$ with $O_i = \sum_j \lambda_{i,j} W_i^\dagger\ket{j}\!\bra{j}W_i$.}
 \KwOut{A flipped evaluation of the conventional model $\Tr[\rho(\x) O(\thv)]$}
 Initialize $o = 0$, $N=\O(\norm{O}_1^2/\eps^2)$\;
 \For{$N$ iterations}{
  Sample $i \in \{1, \ldots, d+1\}$ w.p.~$\left(\frac{w_1\norm{O_1}_1}{\norm{O}_1}, \ldots, \frac{w_d\norm{O_d}_1}{\norm{O}_1}, \frac{\norm{O}_1-\sum_{i=1^d}w_i\norm{O_i}_1}{\norm{O}_1} \right)$\;
  \If{$i = d+1$}{
  break iteration (or alternatively prepare $\sigma(\thv) = I/2^{n+1}$ and jump to line \ref{alg:goto})\;
  }
  Sample $b \in \{+, -\}$ w.p.~$\frac{\norm{O_{i,b}}_1}{\norm{O_i}_1}$\;
  Sample $j \in \{0, \ldots, 2^n-1\}$ w.p.~$\frac{\max(0, b\lambda_{i,j})}{\norm{O_{i,b}}_1}$\;
  Prepare $\sigma(\thv) = |\widetilde{b}\rangle\langle\widetilde{b}| \otimes V^\dagger(\thv) W_i^\dagger\ket{j}\!\bra{j}W_i V(\thv)$, for $\widetilde{b}=2b-1$\;
  Measure $(I\otimes U^\dagger(\x))\sigma(\thv)(I\otimes U(\x))$ in the computational basis and call the outcome $|\widetilde{b}\rangle\otimes\ket{j}$\;\label{alg:goto}
  $o \leftarrow o + b\norm{O}_1\rho_{0,j}$, where $\rho_{0,j}$ is the $j$-th diagonal element of $\rho_0$\;
 }\Return{$o/N$}
\end{algorithm2e}
\end{minipage}
\end{figure}

Out of the observables $O(\thv)$, we need to prepare quantum states $\rho'(\thv)$ such that $\Tr[\rho'(\thv) O'(\x)] = \Tr[\rho(\x) O(\thv)]$. The only difficulty is that quantum states are positive semi-definite and have unit trace while Hermitian observables generally do not fulfill any of these two conditions. To get around these constraints, we simply decompose the observables $O(\thv)$ into positive and negative components, that we both normalize. More precisely, call $O_+(\thv)$ ($O_-(\thv)$) the positive (negative) part of $O(\thv) = O_+(\thv) - O_-(\thv)$. We define:
\begin{equation}
  \begin{cases}
  \rho'_+(\thv) = O_+(\thv)/\norm{O_+(\thv)}_1\\
  \rho'_-(\thv) = O_-(\thv)/\norm{O_-(\thv)}_1
  \end{cases}
\quad\text{and}\quad\quad\quad
  \begin{cases}
  p_+ = \norm{O_+(\thv)}_1/\norm{O(\thv)}_1\\
  p_- = \norm{O_-(\thv)}_1/\norm{O(\thv)}_1
  \end{cases}
\end{equation}
such that
\begin{equation}
\rho'(\thv) = p_+ \ket{0}\!\bra{0} \otimes \rho'_+(\thv) + p_- \ket{1}\!\bra{1} \otimes \rho'_-(\thv)
\end{equation}
is a valid quantum state (positive semi-definite and unit trace). We can then take
\begin{equation}
O'(\x) = \norm{O(\thv)}_1 (\ket{0}\!\bra{0} - \ket{1}\!\bra{1})\otimes \rho(\x)
\end{equation}
which, as one can easily verify, leads to $\Tr[\rho'(\thv) O'(\x)] = \Tr[\rho(\x) O(\thv)]$.\\
However, this still does not give us a proper flipped model as the renormalization factor $\norm{O(\thv)}_1$ of $O'(\x)$ can depend on the parameters $\thv$ (and more precisely the weights $w_i$, see remark above). We would like to use here the upper bound $\norm{O}_1$. To do so, we can simply (re-)define:
\begin{equation}
  \begin{cases}
  p_+ = \norm{O_+(\thv)}_1/\norm{O}_1\\
  p_- = \norm{O_-(\thv)}_1/\norm{O}_1
  \end{cases}
  \quad\text{and}\quad\quad p_0 = \frac{\norm{O}_1-\norm{O_+(\thv)}_1-\norm{O_-(\thv)}_1}{\norm{O}_1}
\end{equation}
and also
\begin{equation}
\begin{cases}
\rho'(\thv) = p_+ \ket{0}\!\bra{0} \otimes \rho'_+(\thv) + p_- \ket{1}\!\bra{1} \otimes \rho'_-(\thv) + p_0 I/2^{n+1}\\
O'(\x) = \norm{O}_1 (\ket{0}\!\bra{0} - \ket{1}\!\bra{1})\otimes \rho(\x)
\end{cases}
\end{equation}
such that we still have $\Tr[\rho'(\thv) O'(\x)] = \Tr[\rho(\x) O(\thv)]$ but where we have now defined a proper flipped model.
\end{proof}

Going a bit further, we also propose an algorithm to evaluate the flipped model we constructed in our proof (see Algorithm \ref{alg:flipping}).
Given that we only assume to know the eigenvalue decomposition of the single $O_i$'s, and not that of the full observable $O(\thv)$, we do not decompose $O(\thv)$ directly into its positive and negative components, but rather the $O_i$'s. For this, we define:
\begin{equation}
  \begin{cases}
  O_{i,+} = \sum_{j, \lambda_{i,j} \geq 0} \lambda_{i,j} W_i^\dagger\ket{j}\!\bra{j}W_i\\
  O_{i,-} = \sum_{j, \lambda_{i,j} \leq 0} \abs{\lambda_{i,j}} W_i^\dagger\ket{j}\!\bra{j}W_i
  \end{cases}
\end{equation}
We then recover $\rho'(\thv)$ via importance sampling of the indices $i,j$ and $\pm$ and the implementation of the pure state $V(\thv) W_i\ket{j}$ (see Algorithm \ref{alg:flipping})\footnote{In Algorithm \ref{alg:flipping}, $\sigma(\thv)$ corresponds to either $\rho'_+(\thv)$ or $\rho'_-(\thv)$, depending on the sampled $b \in \{+, -\}$.}. Naturally, one could alternatively design a full unitary implementation of $\rho'(\thv)$ using auxiliary qubits to prepare coherent encodings of the probability distributions appearing in Algorithm \ref{alg:flipping}, along with controlled operations between these auxiliary qubits and the working register, but this would require more qubits and a more complicated quantum implementation.

	\subsubsection{Lower bounds}
	
\begin{theorem}\label{thm:lower-bound-flipping}
For $d$ an arbitrary positive integer, there exists a conventional model $\Tr[\rho(\x)O(\y)]$, acting on $n \in \O(\log(d))$ qubits, that satisfies $\norm{O(\y)}_1 = d$ for all $\y \in \mathcal{Y}$, such that for any flipped model $\Tr[\rho'(\y)O'(\x)]$ acting on $m$ qubits with $\norm{O'}_\infty = \sup_{\x\in\X} \norm{O'(\x)}_\infty$, and any $\eps \geq 0$, if the model satisfies
\begin{equation}\label{eq:approx-condition}
\abs{\Tr[\rho(\x)O(\y)] - \Tr[\rho'(\y)O'(\x)]} \leq \eps \quad \forall \x,\y \in \X \times \mathcal{Y}
\end{equation}
then, it must also satisfy
$$m\norm{O'}_\infty^2 \in \Omega(d^2(1/2-\eps)^2).$$
\end{theorem}
\begin{proof}
The core of the proof is to show that a conventional model $\text{Tr}[\rho(i) O(\y)]$ with trace norm $\norm{O(\y)}_1 = d$ and acting on $n=\lceil\log_2(N+1)\rceil$ qubits, for $N = \lfloor d^2/4 \rfloor$, can represent the function $i \mapsto y_i$ for $1 \leq i \leq N$, for all $\y\in\{0,1\}^{N}$. For this, we take, for all $1 \leq i \leq N$:
\begin{equation}
\rho(i) = \frac{1}{2}(\ket{0} + \ket{i})(\bra{0} + \bra{i})
\quad \text{and } \quad
O(\y) = \sum_{i'=1}^{N+1} O_{i'}(\y)
\end{equation}
for
\begin{equation}
O_i(\y) =
\begin{cases}
y_{i} (\ket{0}\!\bra{i} + \ket{i}\!\bra{0}) \quad \text{if } 1 \leq i \leq N\\
(d-\sqrt{\abs{\y}}) \ket{N+1}\!\bra{N+1} \quad \text{if } i=N+1
\end{cases}
\end{equation}
where $\abs{\y}=\sum_{i=1}^{N}y_i$ is the Hamming weight of $\y$. By construction, we have that the upper-left $N\times N$ block of $O(\y)$ satisfies $\norm{O_{(N\times N)}(\y)}_1 = \sqrt{\abs{\y}} \leq d$ (it corresponds to the adjacency matrix of a star graph of degree $D = \abs{\y} \leq N$, which has trace norm $2\sqrt{D}$ \cite{wang17}), such that $\norm{O(\y)}_1 = d$ for all $\y \in \{-1,1\}^{N}$. Also, it is easy to check that $\Tr[\rho(i)O(\y)] = y_{i}$ for all $1 \leq i \leq N$.\\
We take this conventional model to be our target model. Satisfying the condition of Eq.~(\ref{eq:approx-condition}) is then equivalent to satisfying:
$$ \abs{\Tr[\rho'(\y)O'(i)] - y_{i}} \leq \eps \quad \forall\ i, \y \in \{1, \ldots, N\}\times\{0,1\}^{N}.$$ 
Now note that this condition is stronger than that of Theorem \ref{thm:lower-bound-shattering} for $\gamma = 1/2-\eps$ and $\alpha_{i,j} = 1/2\ \forall {i,j}$. Therefore, in order not to contradict with this theorem, we must have $m\norm{O'}_\infty^2 \in \Omega(N(1/2-\eps)^2) = \Omega(d^2(1/2-\eps)^2)$.
\end{proof}

Lemma \ref{lem:flipping-bounds} is obtained from Theorem \ref{thm:upper-bound-flipping} as stated above and Theorem \ref{thm:lower-bound-flipping} for the case $m \in \O(\log d)$, such that the statement $m\norm{O'}_\infty^2 \in \Omega(d^2(1/2-\eps)^2)$ becomes $\norm{O'}_\infty \in \Omega\left(\frac{d}{\sqrt{\log(d)}}(1/2-\eps)\right)$.

\subsubsection{Circumventing lower bounds}

Note that our flipping bounds are essentially tight only in the regime where the number of qubits used by the original and resulting model $n,m$ are both in $\O(\text{polylog}(\norm{O}_1))$. This is a relevant regime, as it includes notably the case of Pauli observables (be they local or non-local) or linear combinations thereof. However, outside this regime our bounds can be circumvented.

Also note that an easy way of circumventing our lower bounds, even in the regime where $n,m \in \O(\text{polylog}(\norm{O}_1))$, is by imposing the constraint that $O(\y)$ is parametrized by $\abs{\y}=\O(\text{poly}(n))$ parameters\footnote{In our proof of Theorem \ref{thm:lower-bound-flipping}, we use $\abs{\y}~\in~\Omega(\norm{O}_1^2)$, which is in $\Omega(\text{exp}(n))$ for $n,m \in \O(\log(\norm{O}_1))$, and subexponential in $n$ for $n,m \in \O(\text{polylog}(\norm{O}_1))$}. Indeed, in this case, one can simply use a similar construction to that in Ref.~\cite{jerbi21b} (Fig. 3) where one would encode $\y$ (e.g., in binary form) in auxiliary qubits as $\ket{\widetilde{\y}}$ and use controlled operations that are independent of $\y$ to simulate the action of gates parametrized by $\y$. One can then note that by taking $\rho'(\thv) = \rho_0 \otimes \ket{\widetilde{\y}}\!\bra{\widetilde{\y}}$ and the rest of the resulting circuit to define $O'(\x)$, one ends up with a flipped model that acts on $\O(\text{poly}(n))$ qubits and satisfies $\norm{O'}_\infty = \norm{O}_\infty$. For $O$ defined by a Pauli observable for instance, we have $\norm{O}_\infty=1$ and $\norm{O}_1=2^n$. Such a construction therefore does not suffer from an exploding spectral norm $\norm{O'}_\infty$. However, it also does not lead to shadowfiable models in general as the parametrized states $\rho'(\thv)$ play a trivial role and the observables $O'(\x)$ hide all of the quantum computation.

\section{Quantum advantage using shadow models}\label{appdx:quantum-advantage}

\subsection{Discrete cube root learning task}
In this section we rigorously define the discrete cube root learning task introduced in the main text and detail the proof of its classical hardness. More precisely, we start by introducing the discrete cube root problem and state formally its classical hardness assumption (the discrete cube root assumption). Then we construct a learning task that is classically hard based on this assumption.

\subsubsection{The discrete cube root problem}\label{sec:DCR}
A definition of the discrete cube root problem can be found in Ref.~\cite{kearns94}. For convenience, we restate it in this appendix.\footnote{Note that, as opposed to the exposition of Ref.~\cite{kearns94}, we consider the domain $\Z=\{0, \ldots, N-1\}$ instead of $\{i \;|\; 0<i<N, \; \text{gcd}(i,N)=1\}$ for the functions we define next. This allows us to apply more easily the result of Ref.~\cite{alexi88} and construct a learning task with a stronger form of classical hardness.}
Consider two large prime numbers $p$ and $q$ of the form $3k+2,\ 3k'+2$, for distinct $k,k'$, and which can be represented by approximately the same number of bits. Let $N=pq$ be the product of these primes, which we assume to be an $n$-bit integer, and let $\Z=\{0, \ldots, N-1\}$.

We consider the ``discrete cube" function $f_N(y):\Z\rightarrow \Z$ defined as $f_N(y)=y^3 \text{ mod } N$, as well as its inverse, which we denote as $g_N(x) = f^{-1}_N(x) = \sqrt[3]{x} \text{ mod } N$ (see Fig.~\ref{fig:DCR}). As we explain in the following, $f_N$ is particularly interesting because it is believed to be a \emph{one-way function}: $f_N(y)$ can be computed efficiently classically using modular exponentiation, while $g_N(x)$ is believed hard to compute classically, with only knowledge of $N$ and $x$ (and not the factors $p,q$ of $N$). But first, let us show that the inverse function $g_N$ is properly defined.

\begin{figure}[t]
	\centering
	\includegraphics[width=0.6\linewidth]{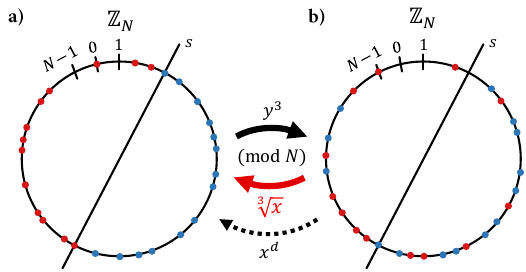}
	\caption{\textbf{A visualization of the functions involved in the quantum advantage learning task.} The core functions of this task map $\Z=\{0,\ldots,N-1\}$ to itself, for $N$ a large semiprime. a) In feature space, data is linearly separable by a hyperplane parametrized by a certain $s\in\Z$. One can efficiently transform data $y$ in feature space into its corresponding data $x$ in input space via the ``discrete cube" function $x = y^3 \text{ mod } N$. b) To a fully classical learner, data in input space looks randomly labeled, as inverting it back to feature space via the discrete cube root function $y = \sqrt[3]{x} \text{ mod } N$ is believed to be classically intractable. However, a shadow model can make use of the trap-door property of the discrete cube root function to efficiently compute a key $d\in\Z$ using a quantum computer and classically map data to feature space through the transformation $y = x^d \text{ mod } N$.}
	\label{fig:DCR}
\end{figure}

\begin{lemma}
    For $p,q$ two distinct prime numbers of the form $3k+2,\ 3k'+2$ and $N=pq$, the function $f_N(y)=~y^3~\text{ mod } N$ is a bijection on $\Z$.
\end{lemma}
\begin{proof}
    To show that $f_N$ is a bijection on $\Z$, it is sufficient to show that it is injective, since it maps $\Z$ to itself.\\
    Consider $y,z\in\Z$ such that $f_N(y)=f_N(z)$, i.e., $y^3 \text{ mod } N = z^3 \text{ mod } N$. We then have $y^3-z^3 \equiv 0 \text{ mod } N$, which means that there exists an $l\in\mathbb{N}$ such that $y^3-z^3 = lN = lpq$. Therefore, we know that $y^3-z^3$ is divisible by both $p$ and $q$, which implies $y^3 \equiv z^3 \text{ mod } p$ and $y^3 \equiv z^3 \text{ mod } q$. From here, we apply a similar reasoning for $p$ and $q$, that we detail only for $p$.\\
    Given that we assumed $p=3k+2$, we know that $\text{gcd}(p-1,3)=1$. Therefore, Euclid's algorithm assures that there exist $d_p,d_p'\geq 1$ such that $3d_p=(p-1)d_p'+1$. Then notice that:
    \begin{equation}\label{eq:y3}
        y^{3d_p} \equiv y^{(p-1)d_p'+1} \equiv y \text{ mod } p
    \end{equation}
    where the last congruence follows from applying Fermat's little theorem ($y^p \equiv y \text{ mod } p$ for all $y\in \mathbb{N}$) to show by induction on $m \geq 0$ that $y^{(p-1)m+1} \equiv y \text{ mod } p$. Similarly, $z^{3d_p} \equiv z \text{ mod } p$, and therefore by raising to the power $d_p$ both terms in $y^3 \equiv z^3 \text{ mod } p$, we get $y \equiv z \text{ mod } p$.\\
    From the same reasoning, we get $y \equiv z \text{ mod } q$, which implies that $x-y$ is divisible by both $p$ and $q$. Since they are distinct primes, then $x-y$ is also divisible by $pq=N$, which means that $x\equiv y \text{ mod } N$. This shows that $f_N$ is injective, and therefore bijective.
\end{proof}

Now that we have shown that the discrete cube root function is properly defined on $\Z$, let us define the discrete cube root problem:

\begin{definition}[Discrete cube root problem]
    Let $p$ and $q$ be two distinct primes of the form $3k + 2, 3k'+2$ represented by approximately the same number of bits, and such that $N=pq$ is an $n$-bit integer. Given as input both $N$ and $x\in \Z$, output $y\in \Z$ such that $y^3=x \text{ mod } N$.
\end{definition}

The assumption that the ``discrete cube" function $f_N$ is a one-way function is formalized by the so-called discrete cube root assumption, which is a special case of the RSA assumption for the exponent $e=3$.

\begin{definition}[Discrete cube root assumption \cite{kearns94}]\label{def:DCR-assumption}
  For any polynomial $P(.)$, there does not exist a classical algorithm  $\A$, that runs in time $P(n)$ and that, on input $N$ and $x$ (where $N$ is the $n$-bit product of two random primes of the form $3k+2$ and $x$ is chosen randomly from $\Z $), outputs $y\in \Z$ such that with probability $1/P(n)$ satisfies $y^3=x \; \text{mod} \; N$. The probability of success is taken over the random draws of the two primes $p$, $q$ and the input $x\in \Z$ and any internal randomisation of $\A$.
\end{definition}

While other one-way functions like the discrete exponential (along with its inverse, the discrete logarithm) also have similar classical hardness assumptions, the discrete cube (root) function is additionally known to be a \emph{trap-door function}. That is, there exists a key $d\in\Z$ such that $g_N(x)$ can be alternatively computed via modular exponentiation as $g_N(x) = x^d \text{ mod }N$. This key $d$ can be efficiently computed from the prime factors $p$ and $q$ of $N$. We show this using a similar reasoning to that around \cref{eq:y3}. Call $\phi(N)=(p-1)(q-1)$. From $p=3k+2$ and $q=3k'+2$, we have $\text{gcd}(\phi(N),3)=1$. Therefore, Euclid's algorithm assures that there exists $d,d'\geq 1$ such that $3d=\phi(N)d'+1 = (p-1)(q-1)d'+1$. We want to show that, for all $y\in\Z$,
\begin{equation}
    y^{3d} \equiv y \text{ mod } N.
\end{equation}
To do this, we show that $y^{3d}-y$ is divisible by both $p$ and $q$. In the case of $p$, we have:
\begin{equation}
    y^{3d} \equiv y^{(p-1)(q-1)d'+1} \equiv y \text{ mod } p.
\end{equation}
The last equality follows again from Fermat's little theorem (see \cref{eq:y3}). Similarly, $y^{3d} \equiv y \text{ mod }q$, which implies that $y^{3d}-y$ is divisible by $p$ and $q$, and that $y^{3d} \equiv y \text{ mod } N$. Therefore $d$ is a valid key for computing $g_N(y)$ for all $y\in \Z$. When knowing the factors $p$ and $q$ of $N$, one can compute $\phi(N)$ and use Euclid's algorithm to find $d$ such that $3d = 1  \text{ mod }\phi(N)$. However, factoring a large $N$ is believed to be computationally intractable classically, which justifies the discrete cube root assumption (Def.~\ref{def:DCR-assumption}).

\subsubsection{The learning task}

Based on the discrete cube root assumption, we can construct a learning task that is not efficiently PAC learnable classically. In order to define the concept class of this learning task, we first consider the class of functions:
\begin{equation}
    \F_n = \{ g_N : \Z\rightarrow \Z\ |\ g_N(x) = \sqrt[3]{x}\; \text{mod}\; N, \text{ for }N\text{ an }n\text{-bit integer that satisfies the DCR conditions}\} 
\end{equation}
defined for any integer $n$. We define the concept class:
\begin{equation}\label{eq:DCR-concepts}
    \C_n = \{ g_{N,s} : \Z\rightarrow \{0,1\}\ |\ g_{N,s}(x) =
    \begin{cases}
    1,&\text{if }g_N(x)\in[s,s+\frac{N-1}{2}],\\ 0, &\text{otherwise.}
    \end{cases}
    , \text{ for }g_N\in\F_n,\ s\in\Z\} 
\end{equation}
for any integer $n$. In our proofs of classical hardness and quantum learnability, we also need that the integer $N$ corresponding to the target function should be specified with the training data. One way of doing so is to append it to all inputs $x \in \Z$ as $(x,N)$ and redefine the concept class accordingly. To ease notation, we assume this transformation to be done implicitly in the following.

Because the discrete cube function is a one-way function and is bijective on $\Z$, it is easy to classically generate training data for any concept $g_{N,s}\in\C_n$, under the uniform distribution over $\Z$. Indeed, one can simply uniformly sample $y\in\Z$, compute its corresponding $x = y^3 \text{ mod }N$ via modular exponentiation and keep from $y$ only the label indicating whether $y\in[s,s+\frac{N-1}{2}]$. The bijectivity of $f_N$ ensures that $x$ is generated uniformly over $\Z$.

In the PAC setting, a learning algorithm has to find, for every concept $g_{N,s}\in\C_n$, for every data distribution $\D$, and for all $\eps,\delta \in (0,1/2)$, a hypothesis function $h$ that, with probability $1-\delta$, satisfies $\Pr_{x \sim \D}[g_{N,s}(x) \neq h(x)]\leq \eps$ in time and number of samples both polynomial in $n$, $1/\eps$ and $1/\delta$.

Note that in Ref.~\cite{kearns94}, the authors consider instead the concept class
\begin{equation}
    \C'_n = \{ g_{N,i} : \Z\rightarrow \{0,1\}\ |\ g_{N,i}(x) = \text{bin}(i,g_N(x)), \text{ for }g_N\in\F_n,\ i\in\{1, \ldots, n\}\} 
\end{equation}
for any integer $n$, where bin$(i,y)$ denotes the $i$-th bit of $y$ in binary form. This concept class is also not efficiently PAC learnable, although this is only shown in a much weaker sense. The authors show that no classical algorithm can achieve an error $\eps=1/n^2$ with failure probability $\delta=1/n^2$ in $\O(\text{poly}(n))$ time, while, for the concept class $\C_n$ we consider, we can show that $\eps=1/2 - 1/\text{poly}(n)$ and $\delta=1/3$ would already break the DCR assumption.

\subsubsection{Classical hardness}

To show that classical learners cannot achieve significantly better than random guesses on the class $C_n$, we make use of a result from Alexi \emph{et al.}\cite{alexi88}. This result makes use of the notion of a $\eps(n)$-oracle:
\begin{definition}[$\eps(n)$-oracle] Let $O_{N,s}$ be a probabilistic oracle, such that $O_{N,s}(x)$ computes $g_{N,s}(x)$ correctly with probability $1/2+\eps(n)$ over the random choice of $x$ and the internal randomness of the oracle. We say that $O_{N,s}$ is an $\eps(n)$-oracle.
\end{definition}
Alexi \emph{et al.} show that a $1/\text{poly}(n)$-oracle is sufficient to break the DCR assumption.
\begin{lemma}[Corollary (a) to Theorem 1 in \cite{alexi88}]\label{lem:alexi}
    For any $g_{N,s}\in\C_n$, given a $1/\text{poly}(n)$-oracle $O_{N,s}$ to $g_{N,s}$, there exists a $\O(\text{poly}(n))$-time algorithm that uses $O_{N,s}$ to compute $g_{N}(x),\ \forall x \in \Z$ (with success probability, e.g., $9/10$).
\end{lemma}
 
We make use of this result to show the following lemma:
\begin{lemma}\label{lem:harndess-DCR}
    Under the discrete cube root assumption, no $\O(\text{poly}(n))$-time classical learning algorithm can achieve an expected error
    \begin{equation}
        \Pr_{x\sim\mathcal{U}(\Z)} [h(x) \neq g_{N,s}(x)] \leq 1/2 - 1/\text{poly}(n)
    \end{equation}
    with probability $2/3$ over the random generation of its training data and its internal randomness, and this for every concept $g_{N,s}\in\C_n$. $\mathcal{U}(\Z)$ is the uniform distribution over $\Z$.
\end{lemma}

\begin{proof}
    Suppose by contradiction that such a learning algorithm would exist for a certain concept $g_{N,s}\in\C_n$. Then, given $N$, one can use this learning algorithm to generate with probability $2/3$ a $1/\text{poly}(n)$-oracle $O_{N,s} = h$. This is possible since the generation of training data for the concept $g_{N,s}$ is classically efficient given $N$. Now, by applying Lemma \ref{lem:alexi}, one obtains a $\O(\text{poly}(n))$-time algorithm that uses this $O_{N,s}$ to compute $g_{N}(x),\ \forall x \in \Z$, with success probability $9/10$. The overall success probability of this procedure, taken over the random choice of $x,\ N$ and the learning algorithm is $0.6$, which contradicts the DCR assumption. 
\end{proof}

\subsection{A simple shadow model}
In this section we show how to construct a simple shadow model which can solve the same learning task for which we just showed classical hardness. This shadow model is obtained from the following flipped model:
\begin{equation}
\begin{gathered}
f_{\thv}(\x)  = \Tr[\rho(\thv)O(\x)]\\
\rho(\thv) = \ket{d',s'}\!\bra{d',s'}\ \&\ O(\x) = \sum_{d',s'} g_{N,s'}(\x) \ket{d',s'}\!\bra{d',s'}.
\end{gathered}
\end{equation}
That is, $\rho(\thv)$ consists of an $n$-qubit register which contains the candidate key $d'\in\Z$ and a second $n$-qubit register containing the candidate separating $s'\in\Z$ that is used to label whether $g_{N}(\x)=\x^d \text{ mod } N \in [s',s'+\frac{N-1}{2}]$.

\begin{lemma}\label{lem:learnability-DCR}
    The concept class $\C_n$ is efficiently PAC learnable under the uniform distribution $\mathcal{U}(\Z)$ using a shadow model.
\end{lemma}
\begin{proof}
    We first describe how $\rho(\thv)$ can be computed efficiently quantumly. Using the specification of $N$ provided by the training data\footnote{Strictly speaking, the PAC framework does not allow one to provide $N$ explicitly in the training data. One way around this issue is to redefine the concepts to be of the form $\tilde{g}_{N,s} : \{0,1\}\times\Z\rightarrow \{0,1\}\ |\ \tilde{g}_{N,s}(0,x)=g_{N,s}(x)$ and $\tilde{g}_{N,s}(1,x)=$ ``$j$-th bit of $N$, for $j$ encoded in the first $\log(N)$ bits of $x$''. This way, we can efficiently recover $N$ from the uniform distribution $\mathcal{U}(\{0,1\}\times\Z)$, while only marginally impacting the training data.}, one can use Shor's algorithm to compute the factors $p,q$ of $N$ with arbitrarily high probability of success. This in turn allows to compute $\phi(N)=(p-1)(q-1)$ and the key $d'= d$ that satisfies $3d = 1  \text{ mod }\phi(N)$ using Euclid's algorithm. As for the candidate separating $s'$, it can be encoded using Pauli-X gates. The observable $O(\x)$ can also be evaluated efficiently from computational basis measurements. For an outcome $(d',s')$, one simply computes $g_{N}(\x)=\x^d \text{ mod } N$ via modular exponentiation and checks whether the output is in $[s',s'+\frac{N-1}{2}]$.\\
    This model is naturally specified as a shadow model. One preparation of $\rho(\thv)$ followed by a computational basis measurement results in the classical advice $\omega(\thv)=(d',s')$. The classical evaluation $\A(\x,\omega(\thv))$ of a new data point $\x$ is done as explained in the last paragraph.\\
    The only remaining learning aspect is to identify an $s'$ close to $s$ from a training set $\{(x,y_i=g_{N,s}(x))\}_{x\sim\mathcal{U}(\Z)}$. We show that a training set $X$ of size $\abs{X} \geq \frac{\log(\delta)}{\log(1-2\eps)}$ is guaranteed to contain an $x^*$ such that, for $s'=g_{N}(x^*)$, $\abs{s-s'}\leq\eps N$ with probability $1-\delta$. We take this $x^*$ to be $x^*=\text{argmin}_{x\in X} \mathcal{L}(x^d \text{ mod } N)$, for $\mathcal{L}(y)= \sum_{x \in X} \abs{g_{N,y}(x) - g_{N,s}(x)}$ the training loss on the training set $X$. We show this by proving:
    \begin{equation}
       \text{Pr}\left(\abs{s'-s} \geq \eps N\right) \leq \delta. 
    \end{equation}
    This probability is precisely the probability that no $g_{N}(x) \in \{g_{N}(x)\}_{x\in X}$ is within $\eps$ distance of $s$, i.e.,
    \begin{equation}
       \text{Pr}\left(\bigcap_{x\in X}g_{N}(x) \notin [s-\eps N, s + \eps N]\right).
    \end{equation}
    As the elements of the training set are all identically distributed, we have that this probability is equal to
    \begin{equation}
    	\text{Pr}\left(g_{N}(x) \notin [s-\eps N, s + \eps N]\right) ^{\abs{X}}.
    \end{equation}
    Since all the datapoints are uniformly sampled from $\Z$, the probability that a datapoint is in any region of size $2\eps N$ is just $2\eps$. With the assumption that $\abs{X} \geq \log_{1-2\eps}(\delta)$ (and assuming $\eps<1/2$), we get:
    \begin{equation}
        \text{Pr}\left(\abs{s'-s} \geq \eps N\right) \leq (1-2\eps)^{\log_{1-2\eps}(\delta/2)} = \delta.
    \end{equation}
    From here, we simply notice that $\abs{s'-s}\leq\eps N$ guarantees an expected error
    \begin{equation}
        \Pr_{x\sim\mathcal{U}(\Z)} [g_{N,s'}(x) \neq g_{N,s}(x)] \leq 2\eps.
    \end{equation}
\end{proof}

\section{Relations between shadow models}

\subsection{Flipped models are universal}\label{appdx:flipped-universal}
In this section we show that any shadowfiable model as defined in Defs.~\ref{def:shadow-model} and \ref{def:shadowfiable} can be approximated by an efficiently shadowfiable flipped model. This result corresponds to Lemma \ref{lem:all-shadowfiable-flipped} in the main text, and more formally in the following lemma. 

\begin{lemma}\label{lem:all-shadowfiable-flipped-extended}
Let $f_{\thv}$ be a shadowfiable model acting on $n$ qubits as defined in Def.~\ref{def:shadowfiable} and let $\eps,\delta > 0$. There exists a flipped model $g_{\thv}(\x) = \Tr[\rho(\thv)O(\x)]$ acting on $\O(\text{poly}(n))$ qubits and evaluatable in $\O(\text{poly}(n,1/\eps))$ time such that
\begin{equation}
	\max_{\x\in\X}\abs{f_{\thv}(\x) - g_{\thv}(\x)} \leq \eps.
\end{equation}
Moreover, this flipped model is also shadowfiable for the error parameter $\eps$ and the success probability $1-\delta$. More precisely, a computational basis measurement of $\rho(\thv)$ yields an advice $\omega(\thv)$ such that with probability $1-\delta$ over the randomness of this measurement, we have
\begin{equation}
	\max_{\x\in\X}\abs{f_{\thv}(\x) - \widetilde{g}_{\thv}(\x)} \leq \eps,
\end{equation}
where $\widetilde{g}_{\thv}(\x) = \A(\x,\omega(\thv))$ is a classical $\O(\text{poly}(n,1/\eps,1/\delta,d))$-time algorithm that processes the advice $\omega(\thv)$ along with an input $\x\in\R^d$.
\end{lemma}

\begin{figure}[t]
	\centering
	\includegraphics[width=0.65\linewidth]{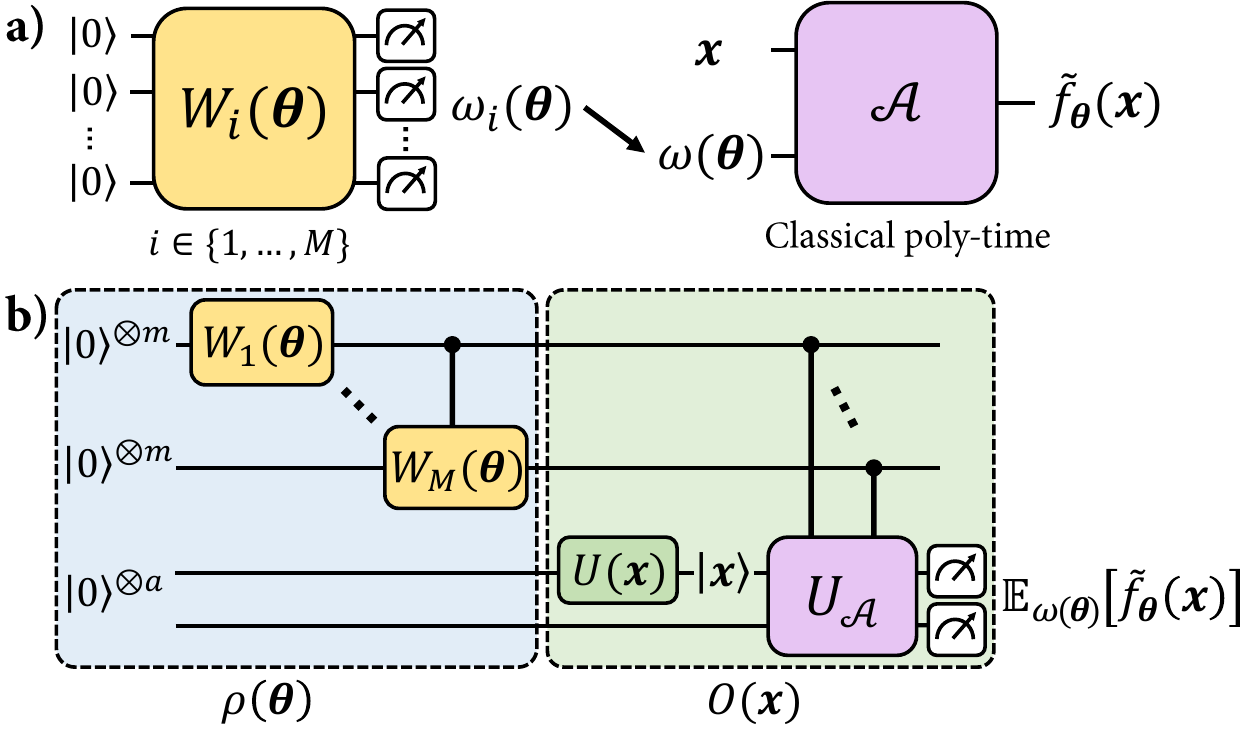}
	\caption{\textbf{All shadow models can be expressed as shadowfiable flipped models}. a) A shadow model consists of $M$ unitary circuits $W_i(\thv)$ that can be chosen adaptively, and that generate advice $\omega_i(\theta)$ from computational basis measurements of the states $W_i(\thv)\ket{0}^{m}$. This advice, along with a (binary description of) an input $\x\in\R^d$ are processed by a classical algorithm $\A$ to compute an approximation $\widetilde{f}_{\thv}$ of the shadowfiable model $f_{\thv}$. b) A coherent implementation of this shadow model, where the unitaries $W_i(\thv)$ are applied on different $m$-qubit registers, and coherently controlled by previous registers (for adaptivity). These $M$ registers constitute the coherent encoding of the advice $\ket{\omega(\thv)}$. The algorithm $\A$ can then be simulated by a reversible quantum computation $U_{\A}$ (see Sec. 3.2.5.~in \cite{nielsen00}) that processes a binary encoding $\ket{\x}$ of $\x$ and the coherent advice $\ket{\omega(\thv)}$ (either directly or indirectly via controlled operations that imprint $\ket{\omega(\thv)}$ on an auxiliary register). This coherent implementation of the shadow model can be viewed as a shadowfiable flipped model $g_{\thv}(\x)=\Tr[\rho(\thv)O(\x)]$, such that one evaluation of this model samples an advice $\omega(\thv)$ and evaluates $\A(\x,\omega(\thv))$ for that advice and a given \x.}
	\label{fig:shadow-flipped}
\end{figure}

\begin{proof}
    By Def.~\ref{def:shadowfiable}, the shadowfiable model $f_{\thv}$ admits for every error of approximation $\eps'>0$ and probability of failure $\delta'>0$ a shadow model $\widetilde{f}_{\thv}$ that uses $m\cdot M\in\O(\text{poly}(n,1/\eps',1/\delta'))$ qubits and guarantees
    \begin{equation}
	   \max_{\x\in\X}\abs{f_{\thv}(\x) - \widetilde{f}_{\thv}(\x)} \leq \eps'
    \end{equation}
    with probability $1-\delta'$ over the generation of its advice. Out of this shadow model, we use the construction described in Fig.~\ref{fig:shadow-flipped}.b) to define the flipped model $g_{\thv}(\x)$. Since this flipped model corresponds to the evaluation of the shadow model $\widetilde{f}_{\thv}(\x)$ averaged over the randomness of $\omega(\thv)$, we have, for all $\x\in\R^d$:
    \begin{align}
        \abs{f_{\thv}(\x) - g_{\thv}(\x)} &\leq \abs{(1-\delta')\eps' + \delta'\norm{\widetilde{f}_{\thv}}}\\
        &\leq \eps' + \delta' \left(\norm{\widetilde{f}_{\thv}}+\eps'\right)
    \end{align}
    where we use that with probability $1-\delta'$ we have $\abs{f_{\thv}(\x) - \widetilde{f}_{\thv}(\x)}\leq\eps'$ and otherwise assume the worse case error $\abs{f_{\thv}(\x) - \widetilde{f}_{\thv}(\x)}\leq 2\norm{\widetilde{f}_{\thv}}$, for $\norm{\widetilde{f}_{\thv}}=\max_{\x\in\X}\abs{\widetilde{f}_{\thv}(\x)}$ (which can also be capped to $\max_{\x\in\X}\abs{f_{\thv}(\x)}$ without loss of generality). Therefore, by setting $\eps'=\frac{\eps}{2}$ and $\delta'=\min\big\{\frac{\eps}{2\left(\norm{\widetilde{f}_{\thv}}+\eps'\right)},\delta\big\}$ (important for the second part of the proof), we get
    \begin{equation}
        \max_{\x\in\X}\abs{f_{\thv}(\x) - g_{\thv}(\x)} \leq \eps.
    \end{equation}
    This proves the first part of the lemma. From here, it is straightforward to notice that a measurement of $\rho(\thv) = \ket{\omega(\thv)}\!\bra{\omega(\thv)} \otimes \ket{0}\!\bra{0}^{\otimes a}$ in the computational basis yields an advice $\omega(\thv)$ such that the algorithm $\A$ associated to the shadow model $\widetilde{f}_{\thv}$ satisfies
    \begin{equation}
	\max_{\x\in\X}\abs{f_{\thv}(\x) - \A(\x,\omega(\thv))} \leq \eps' < \eps
    \end{equation}
    with probability at least $1-\delta' \geq 1-\delta$ over the randomness of measuring the advice.
\end{proof}

\subsection{\textsf{BQP} and \textsf{P/poly}}\label{appdx:BQP-P/poly}
In this section we give a rigorous proof that there exist quantum models that are not shadowfiable, under the assumption that \textsf{BQP} $\not\subset$ \textsf{P/poly}. The approach we take is similar to that of Huang \emph{et al.}~\cite{huang20} (Appendix A), although we consider different complexity classes and therefore show different results. Let us start by noting that the shadow models defined in Def.~\ref{def:shadow-model} compute functions in a subclass of \textsf{P/poly}, namely the subclass in which the advice $\omega(\thv)$ is efficiently generated from the measurements of a polynomial number of quantum circuits. We call this complexity class \textsf{BPP/qgenpoly} and it is obvious that \textsf{BPP/qgenpoly} $\subseteq$ \textsf{P/poly} as the latter is equal to $\textsf{BPP/poly}$ \cite{adleman78} and contains all classically efficiently computable functions with advice of polynomial length, without any constraints on how the advice is generated.

\begin{figure}[t]
	\centering
	\includegraphics[width=0.4\linewidth]{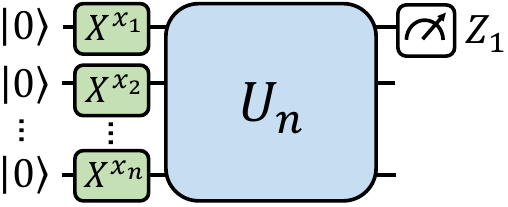}
	\caption{\textbf{A universal quantum model for \textsf{BQP}}. For an $n$-dimensional input $\x\in\{0,1\}^n$, this model acts on $n$ qubits, encodes $\x$ in its binary form $\ket{\x}$ and applies a $\text{poly}(n)$-time unitary $U_n$ before a Pauli-$Z$ measurement of the first qubit. For appropriately chosen unitaries $\{U_n: n\in \mathbb{N}\}$, this model can decide any language in \textsf{BQP}. For more general computational basis measurements, the resulting model can represent arbitrary functions in \textsf{FBQP}, the functional version of \textsf{BQP}.}
	\label{fig:universal-model}
\end{figure}

On the other hand, we know that quantum models can compute all functions in \textsf{BQP}. To see this, we first refer the reader to the definition of \textsf{BQP} in Def.~\ref{def:BQP}. It is easy to show that for any language $L$ in \textsf{BQP}, there exists a quantum model which can decide this language. Consider the following quantum model $f_n=\text{Tr}[\rho(x)O_n]$, depicted in Fig.~\ref{fig:universal-model}:
\begin{equation}
\begin{gathered}
\rho(x) = \bigotimes_{i=1}^n X_i^{x_i}\ket{0}\!\bra{0}X_i^{x_i}\;\; \&\ \;\;O_n = U_n^\dagger Z_1 U_n
\end{gathered}
\end{equation}
Where $X_i$ is the Pauli-X gate acting on the $i$-th qubit, here parametrized by $x_i\in\{0,1\}$, the $i$-th bit of $x$, $Z_1$ is the Pauli observable on the first qubit, and $U_n$ is the quantum circuit used to decide the language $L$ in Def.~\ref{def:BQP}. Then:
\begin{enumerate}
    \item For all $x \in L$, $f_n(x)=$ Pr[the output of $U_{|x|}$ applied on the input $x$ is $1$] - Pr[the output of $U_{|x|}$ applied to the input $x$ is $0$] $\geq 2/3-1/3=1/3$.
    \item For all $x \notin L$, $f_n(x)=$ Pr[the output of $U_{|x|}$ applied on the input $x$ is $1$] - Pr[the output of $U_{|x|}$ applied to the input $x$ is $0$] $\leq 1/3-2/3=-1/3$.
\end{enumerate}
Therefore, as $f_n(x)>0$ if $x \in L$ and $f_n(x)<0$ if $x\notin L$ such quantum model could efficiently decide the language.
We show that if all such quantum models would be shadowfiable then \textsf{BQP} $\subseteq$ \textsf{P/poly}.
\begin{lemma}\label{thm:not-all-shadowfiable-extended}
    If all quantum models $f_{\thv}$ are shadowfiable with the guarantee that, $\forall x \in \X$,
    \begin{equation}
        \abs{f_{\thv}(x)-\widetilde{f}_{\thv}(x)} < 0.15,
    \end{equation}
    with probability at least $2/3$ over the shadowing phase and the randomness of evaluating the shadow model $\widetilde{f}_{\thv}$, then \textsf{BQP} $\subseteq$ \textsf{P/poly}.
\end{lemma}
\begin{proof}
    Consider a language $L$ in \textsf{BQP}. We showed that there exists a quantum model $f_{\thv}=f_n$ which can determine, for every $x\in \{0,1\}^n$, whether $x$ is in $L$. Now, by our assumption, there exists a shadow model $\widetilde{f}_n$ such that, for every $x\in \{0,1\}^n$, $|f_n(x)-\widetilde{f}_n(x)|<0.15$ with probability greater than $2/3$ over the randomness of sampling the advice $\omega(\thv)$ and the (potential) internal randomness of its classical algorithm $\A$. To get rid of these two sources of randomness, we make use of the same proof strategy as for Adleman’s theorem \cite{adleman78}:
    Consider now a new algorithm $\mathcal{A'}$ which runs $\mathcal{A}$ $18n$ times, each time using a new sampled advice string $\omega(\thv)$ and a new random bit-string for its internal randomness. Then we take a majority vote from the $18n$ runs. By Chernoff bound, the probability that for any $x\in \{0,1\}^n$ the algorithm $\mathcal{A'}$ fails to determine if $x$ belongs to $L$ is at most $1/e^n$. Then, by union bound, the probability that $\mathcal{A'}$ decides all the $x\in \{0,1\}^n$ correctly is at least $1-2^n/e^n > 0$. This implies that there exists a particular choice of the $18n$ strings $\omega(\thv)$ and $18n$ random bit-strings used in each run of the algorithm $\mathcal{A}$ such that $\mathcal{A'}$ is correct for all $x$. The algorithm $\mathcal{A'}$, along with this particular choice of $18n$ strings as advice (which is of size polynomial in $n$) is our \textsf{P/poly} algorithm. Note that it guarantees $\forall \x\in \{0,1\}^n$, $\widetilde{f}_n(x)>0$ if $x\in L$ and $\widetilde{f}_n(x)<0$ if $x\notin L$. Therefore we can use the sign of $\widetilde{f}_n(x)$ to determine whether $x\in L$. This implies \textsf{BQP} $\subseteq$ \textsf{P/poly}.
\end{proof}

\subsection{Shadow models beyond Fourier}\label{appdx:beyond-Fourier}

\subsubsection{Generalization to the full domain \texorpdfstring{$\x \in \R^n$}{x in R^n}}

In \cref{sec:beyond-Fourier}, we showed how the model:
\begin{equation}\label{eq:grover-model-appdx}
\begin{gathered}
f_{\y}(\x)  = \Tr[\rho(\x)O(\y)]\\
\rho(\x) = \bigotimes_{i=1}^n R_Y(x_i)\ket{0}\!\bra{0}R_Y^\dagger(x_i)\ \&\ O(\y) = \ket{\y}\!\bra{\y}.
\end{gathered}
\end{equation}
for $\y \in \{0,1\}^{\otimes n}$ is not efficiently shadowfiable when we restrict the domain of $\x$ to be $\{0,\pi\}^n$ as, on this domain, $f_{\y}(\x) = \delta_{\x/\pi,\y}$ plays the role of a database search oracle. We can extend this intractability result to the full domain $\x\in\R^n$ by that noting that more general encoding states (even more general than those in \cref{eq:grover-model-appdx}) take the form
\begin{equation}\label{eq:gen-encoding}
    \rho(\x) = \ket{\psi(\x)}\!\bra{\psi(\x)},\quad \text{for} \quad \ket{\psi(\x)} = \sum_{\bm{j}=1}^{2^n} \alpha_{\bm{j}}(\x) \ket{\bm{j}}
\end{equation}
where $\alpha_{\bm{j}}(\x)\in\mathbb{C}$ is an arbitrary amplitude associated to a computational basis state $\ket{\bm{j}}$.

Now note that when we evaluate the quantum model $f_{\y}(\x)$ on a quantum computer, we do not have direct access to the expectation values $\Tr[\rho(\x)O(\y)]$ it corresponds to. We rather sample an eigenvalue of $O(\y)$ according to the Born rule applied to $\rho(\x)$. Therefore, a single evaluation of $f_{\y}(\x)$ for the encoding $\rho(\x)$ defined in \cref{eq:gen-encoding} returns $1$ with probability $\abs{\alpha_{\y}(\x)}^2$ and $0$ otherwise.

Let us call $U_{\y}$ the Grover operator associated to the database search oracle that marks $\y$, i.e.,
\begin{equation}
    U_{\y} : \ket{\bm{j}}\ket{0} \mapsto \ket{\bm{j}}\ket{\delta_{\bm{j},\y}}.
\end{equation}
and apply it on the state $\ket{\psi(\x)}\ket{0}$:
\begin{equation}
    U_{\y}\ket{\psi(\x)}\ket{0} = \sum_{\bm{j}=1}^{2^n} \alpha_{\bm{j}}(\x) \ket{\bm{j}}\ket{\delta_{\bm{j},\y}}.
\end{equation}
One can then notice that measuring the second register also yields $1$ with probability $\abs{\alpha_{\bm{j}}(\x)}^2$ and $0$ otherwise. Therefore, a single evaluation of $f_{\y}(\x)$ is as powerful as querying the Grover operator oracle in superposition, which still suffers from the same query complexity lower bound as querying it classically.

\subsubsection{Fourier decomposition}

In this subsection, we derive the Fourier-series decomposition of the model in Eq.~(\ref{eq:grover-model-appdx}), i.e., the frequency spectrum $\Omega$ and the Fourier coefficients $c_{\bm{\omega}}$ in the expression:
\begin{equation}
    f_{\y}(\x) = \sum_{\bm{\omega} \in \Omega} c_{\bm{\omega}}(\y) e^{-i\bm{\omega}\cdot\x}.
\end{equation} 
We start by noting that the model has a product structure, in which each $(x_i,y_i)$ pair contributes similarly. Notably, the overlap between the i-th qubit in the state $R_Y(x_i)\ket{0}=\cos(\frac{x_i}{2})\ket{0} - \sin(\frac{x_i}{2})\ket{1}$ and the state $\ket{y_i}$ is: 
\begin{equation}
	\abs{\left(\cos(\frac{x_i}{2})\bra{0} - \sin(\frac{x_i}{2})\bra{1}\right) \ket{y_i}}^2 = \cos^2\left(\frac{x_i + \pi y_i}{2}\right).
\end{equation}
From the Euler decomposition of $\cos(x) = \frac{e^{ix}+e^{-ix}}{2}$, we get:
\begin{equation}
	\cos^2\left(\frac{x_i + \pi y_i}{2}\right) = \frac{2 + e^{i(x_i + \pi y_i)} + e^{-i(x_i + \pi y_i)}}{4},
\end{equation}
such that
\begin{align}
	 f_{\y}(\x) &= \prod_{i=1}^{n} \frac{2 + e^{i(x_i + \pi y_i)} + e^{-i(x_i + \pi y_i)}}{4}\\
	 &= \sum_{\bm{\omega} \in \{-1,0,1\}^n} \frac{1}{2^{n+\abs{\bm{\omega}}}}e^{i (\x+\pi\y) \cdot \bm{\omega}}\\
	 &= \sum_{\bm{\omega} \in \{-1,0,1\}^n} \frac{e^{i\pi \y \cdot \bm{\omega}}}{2^{n+\abs{\bm{\omega}}}}e^{i \x \cdot \bm{\omega}}
\end{align}
where $\abs{\bm{\omega}} = \sum_{i=1}^n \abs{\omega_i}$. We can therefore identify $\Omega = \{-1,0,1\}^n$ and $c_{\bm{\omega}}(\y)= \frac{e^{i\pi \y \cdot \bm{\omega}}}{2^{n+\abs{\bm{\omega}}}}$ and note that the frequency spectrum of this model is exponentially large in $n$, with all its coefficients being non-zero, exponentially small in $n$, and differing only by a phase depending on $\y$. This justifies the exponential sample complexity needed to identify $\y$.

\subsubsection{Obfuscation of \texorpdfstring{$\y$}{y}}

At this point, the interested reader might also point out that the obfuscation of the marked basis state $\ket{\y}\!\bra{\y}$ only occurs somehow artificially by considering the observable $O(\y)$ as a black-box. That is, we consider as a black-box not only the input-independent gates in the circuit but also the mapping from computational basis state to real values when measuring the output state of the circuit. More naturally, when measuring a basis state $\ket{\bm{j}}$, one would have a computable function that returns its corresponding eigenvalue, which in this case could reveal $\y$. An easy fix for this is to include the encoding of $\y$ in the circuit, by redefining $O(\y)$ as 
\begin{equation}
O(\y) = \bigotimes_{i=1}^n X^{y_i}\ket{0}\!\bra{0}X^{y_i},
\end{equation}
which delegates the obfuscation of $\y$ to gates in the circuit.

One can also go a step further in this obfuscation by considering instead the following observables
\begin{equation}
O(\y) = V_{\text{DLP}}\ket{\y}\!\bra{\y}V^\dagger_{\text{DLP}}
\end{equation}
where $V_{\text{DLP}}$ is the unitary\footnote{Note that this is indeed a unitary transformation, but that potentially needs auxiliary qubits to be implemented unitarily on a quantum computer.} that maps a basis state to its discrete logarithm:
\begin{equation}
	V_{\text{DLP}} : \ket{\y} \mapsto \ket{(\log_g(\y) \text{ mop } p) + 1} = \ket{\y'}.
\end{equation}

Now, even the knowledge of $\y$ and a description of the quantum circuit do not help identify $\y'$ classically, under the classical hardness assumption of DLP. Moreover, we still retain the hardness of Fourier-shadowing the resulting model $f_{\y}(\x) = \Tr[\rho(\x)O(\y)]$ from the same database-search arguments. And finally, the flipped model associated to $f_{\y}$ still benefits from the same efficient shadowing procedure, as $O(\y)$ can be prepared on a quantum computer and measured in the computational basis to reveal $\y'$.

\subsection{Shadowfiablility}\label{appdx:shadowfiability}

In our definition of shadowfiable models in the main text (see Def.~\ref{def:shadowfiable}), we take the convention that the shadow model should agree with the original quantum model for all possible inputs $\x\in\X$. This choice makes sense for two reasons:
\begin{enumerate}
    \item We would like the shadowing procedure to work on all potential data distributions, as to be applicable in all learning tasks a given quantum model could be used in.
    \item In the context of machine learning, one typically considers PAC conditions, meaning that the final model should achieve a small error $\E_{\x\sim\D}\abs{h(\x) - g(\x)}$ only with respect to some data distribution $\D$. Note that if the quantum model to be shadowfied achieves these PAC conditions, our demands on worst-case approximation will guarantee that the shadow model achieves them as well.
\end{enumerate}

Nonetheless, one may still be interested in a notion of shadowfiability that considers an average-case error
\begin{equation}
    \E_{\x\sim\D}\big|\widetilde{f}_{\thv}(\x) - f_{\thv}(x)\big|
\end{equation}
with respect to a specified data distribution $\D$. It is not entirely clear which models can still be shadowfied in this way. But our results on the universality of flipped models (Lemma \ref{lem:all-shadowfiable-flipped} in the main text and Lemma \ref{lem:all-shadowfiable-flipped-extended} in the Appendix), as well as on the existence of quantum models that are not shadowfiable (Theorem \ref{thm:not-all-shadowfiable} in the main text and Lemma \ref{thm:not-all-shadowfiable-extended}) would also hold. More precisely, for each of these results, respectively:
\begin{enumerate}
    \item The same proof structure of Lemma \ref{lem:all-shadowfiable-flipped-extended} can be used, as the constructed flipped model only adds a small controllable error to each $\x\in\X$.
    \item One can consider here instead of quantum models that compute arbitrary functions in \textsf{BQP}, a restricted model that computes (single bits of) the discrete logarithm $\log_g(\x) \text{ mod } p$ (analogous to our DCR concept class defined in \cref{eq:DCR-concepts}). The result of Liu \emph{et al.}~\cite{liu20} (Theorem 6 in the Supplementary Information) shows the classical hardness of achieving an expected error $\E_{\x\sim\D}\abs{h(\x) - g(\x)} \leq 1/2 - 1/\text{poly}(n)$ for such target functions (and a hypothesis $h$ producing labels $h(\x)\in\{0,1\}$), under the assumption that \textsf{DLP} $\not\in$ \textsf{BPP}. One can then use this result to show that there exist quantum models that are not average-case shadowfiable under the assumption that \textsf{DLP} $\not\in$ \textsf{P/poly}. 
\end{enumerate}
\end{document}